\documentclass[journal]{IEEEtran}
\IEEEoverridecommandlockouts
\usepackage{amsmath,amssymb,amsfonts}
\usepackage{array}
\usepackage{algorithmic}
\usepackage{graphicx}
\usepackage{float}
\usepackage[caption=false,font=footnotesize]{subfig}
\usepackage{wrapfig}

\usepackage{stfloats}

\usepackage{textcomp}

\newlength{\xfigwd}
\usepackage{caption}

\usepackage[hyphens]{url}
\usepackage{bbm}
\usepackage[algoruled, linesnumbered]{algorithm2e} 
\usepackage{mathrsfs}
\usepackage{mdwmath}
\usepackage{mdwtab}
\usepackage{amsthm}
\usepackage[margin=0.7in]{geometry}

\usepackage{cite}

\newtheorem{lemma}{Lemma}

\newtheorem{remark}{Remark}
\newtheorem{proposition}{Proposition}
\DeclareMathOperator*{\argmax}{argmax}
\DeclareMathOperator*{\argmin}{argmin}
\usepackage{amsmath,amssymb}

\newcommand{\FRF}{\mathbf{F}_{\mathrm{RF}}}
\newcommand{\FBB}{\mathbf{F}_{\mathrm{BB}}}
\newcommand{\fBBk}{\mathbf{f}_{\mathrm{BB},k}}

\newcommand{\figref}[1]{\figurename~\ref{#1}}

\begin{document}

\title{Joint Phase-Time Arrays: A Paradigm for Frequency-Dependent Analog Beamforming in 6G}
\author{Vishnu V. Ratnam, \IEEEmembership{Senior Member,~IEEE}, Jianhua Mo, \IEEEmembership{Senior Member,~IEEE}, Ahmad AlAmmouri, \\ \IEEEmembership{Member,~IEEE},  Boon L. Ng, \IEEEmembership{Member,~IEEE}, Jianzhong (Charlie) Zhang, \IEEEmembership{Fellow,~IEEE}, \\ Andreas F. Molisch, \IEEEmembership{Fellow,~IEEE}
\thanks{V. V. Ratnam (e-mail: vishnu.r@samsung.com, ratnamvishnuvardhan@gmail.com), J. Mo, A. AlAmmouri, B. L. Ng, and J. Zhang are with the Standards and Mobility Innovation Lab, Samsung Research America, Plano, Texas, USA. A. F. Molisch is with the Ming Hsieh Dept. of Electrical and Computer Engg., University of Southern California, Los Angeles, USA.}}



\markboth{This paper has been published in IEEE Access 2022, doi: 10.1109/ACCESS.2022.3190418}
{This paper is published in IEEE Access, doi: 10.1109/ACCESS.2022.3190418}

\maketitle

\begin{abstract}
Hybrid beamforming is an attractive solution to build cost-effective and energy-efficient transceivers for millimeter-wave and terahertz systems. However, conventional hybrid beamforming techniques rely on analog components that generate a frequency flat response such as phase-shifters and switches, which limits the flexibility of the achievable beam patterns. As a novel alternative, this paper proposes a new class of hybrid beamforming called Joint phase-time arrays (JPTA), that additionally use true-time delay elements in the analog beamforming to create frequency-dependent analog beams. Using as an example two important frequency-dependent beam behaviors, the numerous benefits of such flexibility are exemplified. 
Subsequently, the JPTA beamformer design problem to generate any desired beam behavior is formulated and near-optimal algorithms to the problem are proposed. Simulations show that the proposed algorithms can outperform heuristics solutions for JPTA beamformer update. Furthermore, it is shown that JPTA can achieve the two exemplified beam behaviors with one radio-frequency chain, while conventional hybrid beamforming requires the radio-frequency chains to scale with the number of antennas to achieve similar performance. Finally, a wide range of problems to further tap into the potential of JPTA are also listed as future directions.
\end{abstract}


\section{Introduction} \label{sec_intro}
Due to the rising demand for traffic, wireless systems are moving towards higher frequency of operation, such as millimeter-wave (mm-wave) and terahertz (THz) frequencies, where abundant spectrum is available \cite{Boccardi2014}. However, the higher frequencies also suffer from a high isotropic channel propagation loss, and therefore require a large antenna array to create sufficient beamforming gain to ensure sufficient link budget for operation. Thus, these high frequency systems are usually built with a massive antenna array at the transmitter (TX) and/or the receiver (RX) containing many individual antenna elements. At the operating bandwidths of these mm-wave and THz systems, the cost and power consumption of radio-frequency (RF) chain components such as analog-to-digital converters (ADCs) and/or digital-to-analog converters (DACs) also grows tremendously \cite{Murmann_ADC_compiled}. Thus fully digital transceiver implementations, where each antenna element of the massive array is fed by a dedicated RF chain, are very challenging. To keep the hardware cost and power consumption of such massive antenna arrays manageable, typically a \emph{hybrid beamforming} (HBF) architecture is adopted where the antenna array is fed with a much smaller number of RF chains via the use of analog hardware such as phase-shifters and switches \cite{Molisch_VarPhaseShift, Molisch_HP_mag, Heath2016, Alkhateeb_COMM14}. 

A plethora of HBF architectures have been proposed in literature that use as analog hardware: phase-shifters \cite{Sudarshan2006, Ayach_TWC14, Liu2014, Vishnu_jrnl1, Yu2016}, switches \cite{Molisch2004, Vishnu_Globecom, Rial2016}, a combination of phase-shifters and switches \cite{Molisch_VarPhaseShift, Vishnu_ICC2017, Ratnam_HBwS_jrnl}, Rotman lens \cite{Zeng2014, Gao2017} etc. A variety of ways of connecting the analog hardware to the antenna elements have also been explored such as fully-connected \cite{Sudarshan2006, Ayach_TWC14, Sohrabi2016}, partially-connected \cite{Shuangfeng2015}, dynamic architectures \cite{Park2017a, Ratnam_HBwS_jrnl} etc. In such architectures, a combination of beamforming in the digital-domain using RF chains, and in the analog-domain (using phase-shifters or switches) are used to create the overall desired beam-shape in the desired direction. One of the drawbacks of such conventional HBF architectures is that the analog phase-shifters and switches have a frequency-flat response \footnote{Strictly speaking, the beam squint phenomenon in the HBF is also frequency-dependent. However, the angular dispersion in the frequency domain due to the beam squint is relatively small and, more importantly, cannot be adaptively controlled which is the requirement considered in this paper.}, i.e., all components of the input signal frequency undergo a similar transformation after passing through them. This reduces the flexibility of the effective beam-shape that can be created using HBF, in comparison to fully-digital implementations where each antenna array is fed with a dedicated RF chain. For example, in HBF with 1 RF chain, all signal frequencies experience the same beam shape, and thus serving different locations on different frequency sub-bands with good beamforming gain is infeasible. The impact of this drawback is further exacerbated at the mm-wave and THz frequencies, where the frequency-dependent beam shapes are beneficial for many applications such as coverage extension, efficient uplink control channel realization, beam-tracking, initial access, frequency multiplexing of users etc, as shall be detailed in Section \ref{sec_des_beam_behave}.\footnote{Conventional HBF architectures can still realize frequency-dependent beam variation using digital beamforming. However this needs multiple RF chains (as shall be shown in Section \ref{sec_compare_hybrid_BF}), which increases cost and power consumption.} 

In this work, we propose a new class of hybrid architectures called \emph{Joint phase-time arrays} (JPTA) that uses, in addition to the aforementioned frequency-flat analog hardware, true-time delay (TTD) elements, to perform the analog beamforming. Unlike switches and phase-shifters, TTDs can realize a  frequency-dependent phase-shift, i.e., different components of the input signal frequency undergo different transformations after passing through them. This allows the effective analog beam shape created by JPTA to vary with the frequency of the signal, thus removing the draw-back of conventional HBF of only realizing frequency-flat analog beams. As shall be shown in Section \ref{sec_des_beam_behave} this can be beneficial for many applications. In recent times, significant progress has been made in the implementation of tunable TTDs \cite{Integrated_Hashemi2008, True_Rotman2016} that use either optical components \cite{Optically_Frigyes1995, RF_Xie2021} or CMOS components operating at base-band \cite{A_Jang2019, An_Ghaderi2019, A_Lin2021}, cm-wave \cite{True_Chu2013, A_Spoof2020} or mm-wave \cite{Silicon_Ma2015} frequencies. From a system perspective, the use of TTDs in wide-band multi-antenna systems to prevent beam-squinting effects by creating unidirectional, frequency-flat beams was explored in \cite{Delay_Jingbo2019, Delay_Dai2021, Dynamic_Yan2021, Optical_Liu2021}. The use of TTDs to reduce beam training overhead in multi-antenna systems by creating a \emph{rainbow} beam-pattern is considered in \cite{Wideband_Han2019, Design_Boljanovic2020, True_Zhou2021, Fast_Boljanovic2021, Compressive_Boljanovic2021}. Finally, the use of TTDs for low-latency channel access by creating a \emph{rainbow} beam-pattern is considered in \cite{THzPrism_Zhai2020, Rainbow_Li2022}, and the use of multiple RF chains to create multiple such rainbow beam-patterns is considered in \cite{SS_Zhai2021}. 

Although the aforementioned works have considered several different use cases of TTDs in multi-antenna systems, they focus on realizing a specific type of beam-pattern, e.g. a uni-directional beam-pattern or a rainbow beam-pattern, and use a specific TTD architecture. The use of TTDs to realize a new class of hybrid beamforming architectures, viz. JPTA, that can realize different desirable frequency-dependent beam variations has not been explored before. Furthermore, a generic architecture for connecting the TTDs to the antenna elements and the generic beamformer design problem have not been explored before.  The contributions of this paper are as follows:
\begin{itemize}
\item We propose a new type of architecture (JPTA) that can realize frequency-dependent analog beamforming.
\item We formulate the JPTA beamformer design problem to adapt the TTDs and phase-shifters to achieve the desired frequency-dependent analog beams, and propose an algorithm to find near-optimal solutions to this problem.
\item We pick two important desirable beam behaviors that have many use cases, and present good heuristic solutions to the beamformer design problem for these behaviors to serve as baseline.
\item We perform a detailed investigation on the ability of JPTA to achieve these beam behaviors under various system parameter settings. 
\item We also perform an investigation of how many RF chains are required by conventional HBF solutions to replicate similar beam behavior as JPTA with one RF chain, and compare the resulting hardware costs. 
\item Finally, we also enumerate other potential use-cases and future directions for the proposed JPTA architecture.
\end{itemize}
The organization of the paper is as follows: benefits of frequency-dependent beam variations are exemplified in Section \ref{sec_des_beam_behave}, the system model is discussed in Section \ref{sec_sys_model}; the optimal design of the analog hardware is formulated in \ref{sec_prob_formulate}; the proposed algorithm is described in Section \ref{sec_proposed_algo} and some baseline heuristics for comparison are proposed in Section \ref{sec_baseline_algos}; the efficacy of JPTA to achieve different beam behaviors under varying system parameters is discussed in Section \ref{sec_simulation_results} and the comparison to conventional HBF is listed in Section \ref{sec_compare_hybrid_BF}; the future directions are discussed in Section \ref{sec_future_dir} and the conclusions are summarized in Section \ref{sec_conclude}.

\textbf{Notation:} scalars are represented by light-case letters; vectors and matrices by bold-case letters; and sets by light-case calligraphic letters. Additionally, ${\mathrm{j}} = \sqrt{-1}$, ${\mathbf{A}}^{\mathrm{T}}$ is the transpose of a matrix $\mathbf{A}$, ${\mathbf{A}}^{*}$ is the element-wise complex conjugate of a matrix $\mathbf{A}$, ${\mathbf{A}}^{\dag}$ is the Hermitian transpose of a matrix $\mathbf{A}$ (i.e., ${\mathbf{A}}^{\dag} = {[{\mathbf{A}}^{\mathrm{T}}]}^{*}$).  The element in $i$-th row and $j$-th column of the matrix $\mathbf{A}$ is denoted by $\left[\mathbf{A}\right]_{i,j}$.
Furthermore, ${[\mathbf{a}]}_{m}$ represents the $m$-th element of a vector $\mathbf{a}$, $\|\mathbf{a} \|$ represents the $\ell_2$ norm of a vector $\mathbf{a}$, $e^{\mathbf{a}}$ represents the element-wise exponentiation of a vector $\mathbf{a}$, $\angle {\mathbf{a}}$ represents the vector of element-wise phase angles of a vector $\mathbf{a}$. Furthermore, 
$\mathrm{Re}\{\cdot\}$ and $\mathrm{Im}\{\cdot\}$ refer to the element-wise real and imaginary components, respectively, of a vector or scalar, $\mathbb{N}$ is the set of natural numbers and ${\mathrm{O}}(a)$ is the big-O notation implying that the parameter scales linearly with $a$. 

\section{Desired beam-behaviors and their applications} \label{sec_des_beam_behave}
In this section, we consider two important frequency-dependent beam-behaviors that have multiple applications for mm-wave and THz systems. Other potentially useful behaviors are discussed briefly later in Section \ref{sec_future_dir}. Note that a `similar effect' to any frequency-dependent beamforming can be realized with conventional frequency-flat beamforming by time-multiplexing over multiple frequency-flat beams. However, this time multiplexing can have a large overhead, can suffer from a larger latency and can also be spectrally-inefficient, as discussed in the example scenarios below. 

\subsection{Behavior 1} In this behavior, at any signal frequency $f$, the desired beam creates the maximum possible array-gain in one angular direction $\theta(f)$. As $f$ varies linearly over the system bandwidth, the angular direction $\theta(f)$ also sweeps linearly over a certain angular region $[\theta_0 - \Delta \theta/2, \theta_0 + \Delta \theta/2]$ as shown in Fig.~\ref{Fig_behave1}.  
This beam behavior can be useful at a BS to reduce the communication overhead, in scenarios where many users are uniformly distributed in an angular region and require simultaneous channel access for low-latency and low-data-rate service in uplink and/or downlink. An example is the physical uplink control channel (PUCCH) where all the users that require channel access, to send HARQ-ACK (hybrid automatic repeat request acknowledgement) and channel-state information (CSI) reports, can be served simultaneously on different portions of the bandwidth. By avoiding the need to time-multiplex beams to serve these users, this can significantly reduce the PUCCH overhead. Additionally, if this beam behavior is used for the physical downlink shared channel (PDSCH), the required CSI feedback from each user in the PUCCH is also reduced to that of one sub-band, further reducing PUCCH overhead. 
This beam behavior is also useful for uplink coverage extension -- a crucial requirement for mm-wave and THz systems. By allowing simultaneous PUCCH and physical uplink shared channel (PUSCH) access to multiple cell-edge users on different sub-bands, each user can transmit uplink signals for a longer time, thus accumulating more signal power to improve signal-to-noise ratio, leading to cell coverage extension \cite{XDD_samsung}. 
Finally, since such a behavior can enable a BS to maintain a large beamforming gain over a wide region on different sub-bands, it is beneficial in beam-tracking and maintaining good link reliability for a high velocity users, and also in significantly reducing the overhead of cell discovery and initial beam alignment, as considered in \cite{Wideband_Han2019, Fast_Boljanovic2021} where it is referred to as rainbow beamforming. 
\subsection{Behavior 2} In this behavior, again, at any signal frequency $f$, it is desirable for the beam to create the maximum possible array-gain in one angular direction $\theta(f)$. However, here $\theta(f) = \theta_1$ for the lower half of the signal bandwidth and another angle $\theta(f) = \theta_2$ for the upper half bandwidth, as shown in Fig.~\ref{Fig_behave2}. 
This behavior is useful in scenarios where the users are sparsely distributed in the angular domain, and we wish to provide service to multiple users simultaneously on different portions of the large available system bandwidth. Such a behavior is also useful, for example, when the BS supports a large system bandwidth covering two sub-bands, while each user device only supports operation on one sub-band at a time, thus requiring scheduling of different users on the lower and upper sub-bands. Other reasons may include location-specific bandwidth availability to mitigate interference to incumbents, exploiting multi-user diversity to counter frequency-dependent fading etc.
\begin{figure}[!htb]
\centering
\subfloat[Behavior 1]{\includegraphics[width= 0.45\textwidth]{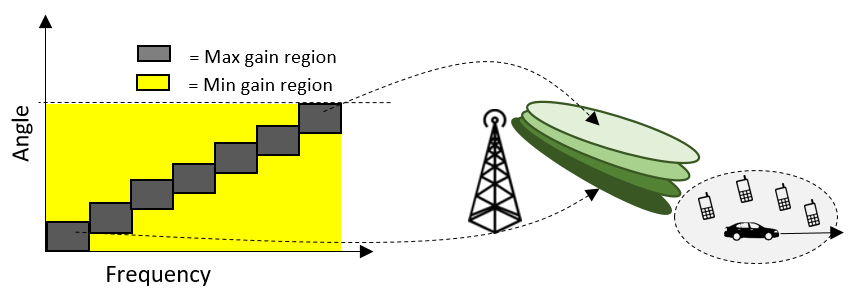} \label{Fig_behave1}} \\ 
\subfloat[Behavior 2]{\includegraphics[width= 0.45\textwidth]{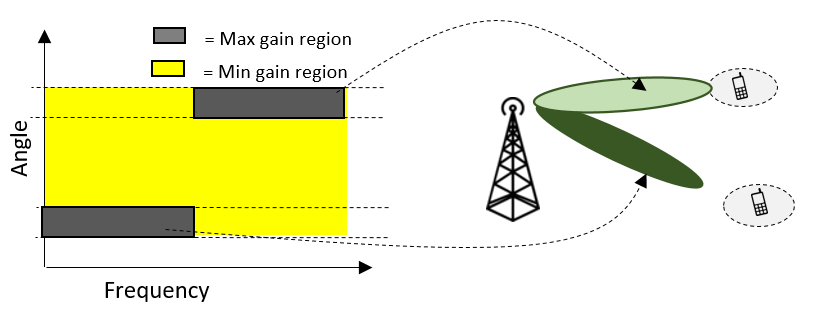} \label{Fig_behave2}}
\caption{An illustration of the desired frequency-dependent beam shapes and their applications.}
\label{Fig_desired_beam_behavior}
\end{figure}

\section{System model} \label{sec_sys_model}
We consider a layout with a single base-station (BS) serving many users in its coverage area and operating with a system bandwidth $W$ around a center frequency $f_0$. The BS is assumed to have a uniform linear antenna array having $M$ elements, and only $N_{RF}=1$ RF chain.\footnote{The presented results can be directly extended to planar array configurations. Extensions that consider multiple RF chains will be explored in future work}. The antenna spacing is half-wavelength at the center frequency $f_0$. Each of the $M$ antennas has a dedicated phase-shifter, and they are connected to the single RF chain via a network of $N \leq M$ TTDs as shown in Fig.~\ref{Fig_BS_architecture}. Here $\mathbf{P}$ is a fixed $M \times N$ mapping matrix, where each row $m$ has exactly one non-zero entry and determines which of the $N$ TTDs antenna $m$ is connected to. The TTDs are assumed to be configurable, with a delay variation range of $0 \leq \tau \leq \kappa/W$, where $\kappa$ is a design parameter to be selected. The phase-shifters are assumed to have unit magnitude and have arbitrarily reconfigurable phase $-\pi \leq \phi < \pi$. Transmission in both uplink and downlink directions is performed using OFDM with $K$ subcarriers indexed as $\mathcal{K} = \{\lfloor \frac{1-K}{2} \rfloor,..., \lfloor \frac{K-1}{2} \rfloor \}$. Then, the $M \times 1$ downlink TX signal on sub-carrier $k \in \mathcal{K}$ for a representative OFDM symbol can be expressed as:
\begin{align}
\mathbf{x}_k &= \underbrace{\frac{1}{\sqrt{M}} 
\begin{bmatrix}
    e^{{\mathrm{j}} \phi_{1}} & 0 & \hdots & 0 \\
    0 & e^{{\mathrm{j}} \phi_{2}} & \hdots & 0 \\
    \vdots & \vdots & \ddots & \vdots \\
    0 & 0 & \hdots & e^{{\mathrm{j}} \phi_{M}} 
  \end{bmatrix} }_{\mathbf{T}}  \mathbf{P}
\underbrace{\left[\begin{array}{c} e^{-{\mathrm{j}} 2 \pi f_k \tau_{1}} \\ e^{-{\mathrm{j}} 2 \pi f_k \tau_{2}} \\ \vdots \\ e^{-{\mathrm{j}} 2 \pi f_k \tau_{N}}\end{array} \right]}_{\mathbf{d}_k} \nonumber \\
& \quad \times \alpha_k s_k \nonumber \\
&= \mathbf{T} \mathbf{P} \mathbf{d}_{k} \alpha_k s_k, \label{eqn_tx_signal}
\end{align}
where $s_k$ and $\alpha_k$ are the scalar data and digital beamforming on the $k$-th subcarrier, $f_k$ is the frequency of the $k$-th sub-carrier (including the carrier frequency), $\tau_n$ is the delay of the $n$-th TTD and $\phi_m$ is the phase of the $m$ phase-shifter connected to the $m$-th antenna. Note that from \eqref{eqn_tx_signal} the total transmit power of the BS can be given by $P_{\mathrm{sum}} = \sum_{k \in \mathcal{K}} {|\alpha_k|}^2$. Note that for this JPTA architecture, the effective downlink unit-norm analog beamformer on sub-carrier $k$ is $\mathbf{T} \mathbf{P} \mathbf{d}_{k}$, where the $M \times M$ diagonal matrix $\mathbf{T}$ captures the effect of phase-shifters and the $N \times 1$ vector $\mathbf{d}_k$ captures the effect of TTDs. It can be shown that the same beamformer is also applicable at the BS for uplink scenario. 
For a target desired beam-behavior, we shall also assume knowledge of the desired $M \times 1$ beamforming vector as a function of sub-carrier index as: ${\boldsymbol{\mathcal{B}}} = \{ {\mathbf{b}}_k | k \in \mathcal{K} \}$, that satisfies $\sum_{k \in \mathcal{K}} {\|\mathbf{b}_k\|}^2 \leq P_{\mathrm{sum}}$. For example, for beam behavior 1 in Section \ref{sec_des_beam_behave}, this desired target beam is: ${\mathbf{b}}_k = \sqrt{P_{\mathrm{sum}}/MK} \mathbf{a}_k\big(\theta_0 + k \Delta \theta/K\big)$ and for beam behavior 2, we have ${\mathbf{b}}_k = \sqrt{P_{\mathrm{sum}}/MK} \mathbf{a}_k\big(\theta_1\big)$ for $k < 0$ and ${\mathbf{b}}_k = \sqrt{P_{\mathrm{sum}}/MK} \mathbf{a}_k\big(\theta_2\big)$ for $k \geq 0$, where $\mathbf{a}_k(\theta)$ is the array response vector at the BS for a given angle $\theta$ given by:
\begin{align}
\mathbf{a}_k(\theta) = {\left[\begin{array}{cccc}
1 & e^{{\mathrm{j}} \frac{\pi \sin(\theta)f_k}{f_0}} & \hdots & e^{{\mathrm{j}} \frac{(M-1) \pi \sin(\theta)f_k}{f_0}}
\end{array}\right]}^{\mathrm{T}}.
\end{align}
Note that here the array response vector $\mathbf{a}_k(\theta)$ is a function of $k$ since we do not ignore the beam-squint effects \cite{An_Ghaderi2019, Delay_Jingbo2019, Dynamic_Yan2021} which can be significant in wide-band systems. 
%
\begin{figure}[!htb]
\centering
\includegraphics[width= 0.45\textwidth]{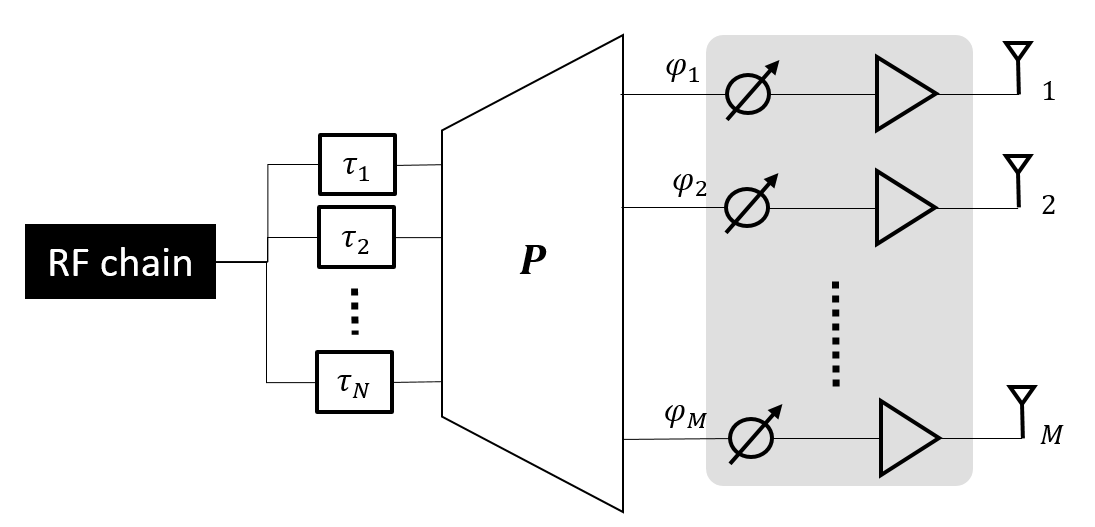}
\caption{An illustration of the BS JPTA architecture with $1$ RF chain and single phase-shifter per antenna element.}
\label{Fig_BS_architecture}
\end{figure}

\section{Problem formulation} \label{sec_prob_formulate}
For a given set of JPTA parameters $\boldsymbol{\alpha} \triangleq \{\alpha_k | k \in \mathcal{K}\}$, $\boldsymbol{\tau} \triangleq \{\tau_n | 1 \leq n \leq N \}$, $\boldsymbol{\phi} \triangleq \{\phi_m | 1 \leq m \leq M \}$, we quantify the matching to the desired beamformer $\mathcal{B}$ as:
\begin{align} \label{eqn_tildef_obj_defn}
\widetilde{\mathscr{F}}_{\mathrm{obj}} (\boldsymbol{\alpha}, \boldsymbol{\tau}, \boldsymbol{\phi}, \boldsymbol{\mathcal{B}}) = \sum_{k \in \mathcal{K}} \frac{1}{K} \Bigg[ {(\|\mathbf{b}_k\| - |\alpha_k|)}^2 \nonumber \\
+ \omega_k {\left\| \frac{\mathbf{b}_k}{\|\mathbf{b}_k\|} - \mathbf{T P}\mathbf{d}_k e^{{\mathrm{j}} \angle \alpha_k} \right\|}^2 \Bigg], 
\end{align}
where $\omega_k$ is an arbitrary non-negative, sub-carrier weighting function. Some candidate choices for $\omega_k$ include: $\omega_k=1$, $\omega_k={\|\mathbf{b}_k\|}^2$, $\omega_k={\|\mathbf{b}_k\|}^2\big/ {\big(1 + {\|\mathbf{b}_k\|}^2 \big)}$, etc. 
In \eqref{eqn_tildef_obj_defn}, the first term of $\widetilde{\mathscr{F}}_{\mathrm{obj}}(\cdot)$ takes care of matching the desired power allocation, while the second term is used to design the analog precoding. Note that this objective is different from the least squares metric: $\sum_{k \in \mathcal{K}} {\left\| \mathbf{b}_k - \mathbf{T P}\mathbf{d}_k \alpha_k \right\|}^2/K$ used in conventional HBF design \cite{Ayach_TWC14}, and is chosen since it achieves better fairness across the sub-carriers.
Noting that $\|\mathbf{T P}\mathbf{d}_k\| = 1$, the design of the JPTA analog beamformer that minimizes $\widetilde{\mathscr{F}}_{\mathrm{obj}}(\cdot)$ is given by:
\begin{flalign}
& \boldsymbol{\tau}^{\circ}, \boldsymbol{\phi}^{\circ}, \angle \boldsymbol{\alpha}^{\circ} &
\nonumber \\
& \quad = \argmax_{\boldsymbol{\tau}, \boldsymbol{\phi}, \angle \boldsymbol{\alpha}} \left\{ \sum_{k \in \mathcal{K}} \omega_k {\mathrm{Re}} \left[ e^{{\mathrm{j}} \angle \alpha_k} \bar{\mathbf{b}}_k^{\dag} \mathbf{T} \mathbf{P} \mathbf{d}_{k} \right] \right\}, \!\!\!\!\!\!\!\! & \label{eqn_opt_prob_analog}
\end{flalign}
where we define $\bar{\mathbf{b}}_k \triangleq {\mathbf{b}}_k/\|{\mathbf{b}}_k\|$ and $\angle \boldsymbol{\alpha} \triangleq \{\angle \alpha_k | k \in \mathcal{K}\}$ and use $\circ$ to indicate optimality of the result. In the rest of the paper, without loss of generality, for \eqref{eqn_opt_prob_analog} we shall use the search ranges $-\frac{\kappa}{2W} \leq \tau_n \leq \frac{\kappa}{2W}$, $-\pi \leq \phi_m \leq \pi$ and $-\pi \leq \angle\alpha_k \leq \pi$. The final TTD values in range $[0,\kappa/W]$ can then simply be obtained by using $\tau_n \leftarrow \tau_n - \min_{1 \leq \bar{n} \leq N} \{\tau_{\bar{n}}\}$ -- a step which doesn't affect optimality and which shall be assumed to be implicitly performed at the end. 
Similarly, the optimal JPTA digital power allocation $|\boldsymbol{\alpha}| \triangleq \{|\alpha_k| | k \in \mathcal{K} \}$ that minimizes $\widetilde{\mathscr{F}}_{\mathrm{obj}}(\cdot)$ is given by:
\begin{align}
{|\boldsymbol{\alpha}|}^{\circ} = \argmin_{|\boldsymbol{\alpha}|} \left\{ \sum_{k \in \mathcal{K}} \frac{1}{K} {\left(\|\mathbf{b}_k\| - |\alpha_k|\right)}^2 \right\}, \label{eqn_opt_prob_digital}
\end{align}
subject to $\sum_{k \in \mathcal{K}} {|\alpha_k|}^2 \leq P_{\mathrm{sum}}$, where we use $\circ$ to indicate the optimality of the result.\footnote{Note that if the channel coefficients are known, a final update to this digital power allocation can be made based on the the analog precoding, for e.g., via water-filling.} 

\section{Digital and analog precoder design} \label{sec_proposed_algo}
Since the desired beamformer set $\{\mathbf{b}_k | k \in \mathcal{K}\}$ satisfies the power constraint, $\sum_{k \in \mathcal{K}} {\|\mathbf{b}_k\|}^2 \leq P_{\mathrm{sum}}$, the optimal solution to \eqref{eqn_opt_prob_digital} is ${|\alpha_k|}^{\circ} = \|{\mathbf{b}}_k\|$. 
Unfortunately, however, finding the joint optimal solution to $\boldsymbol{\phi}, \boldsymbol{\tau}, \angle\boldsymbol{\alpha}$ that maximizes \eqref{eqn_opt_prob_analog} is intractable. Therefore, we consider an alternating optimization approach where we find the conditionally optimal $\boldsymbol{\phi}, \boldsymbol{\tau}$ for a given $\angle\boldsymbol{\alpha}$ and vice-versa. By iterating over these optimizations for a few steps, the JPTA approximation to the desired beamformer can be obtained. 
For a given set of digital-beamforming phases $\angle\boldsymbol{\alpha}$ the optimal JPTA analog precoder can be obtained as:
\begin{flalign}
& \quad \boldsymbol{\tau}^{\circ}(\angle\boldsymbol{\alpha}), \boldsymbol{\phi}^{\circ}(\angle\boldsymbol{\alpha}) = \argmax_{\boldsymbol{\tau}, \boldsymbol{\phi}} \bigg\{ & \nonumber \\
& \qquad \qquad \sum_{k \in \mathcal{K}} \omega_k {\mathrm{Re}} \left[ e^{{\mathrm{j}} \angle \alpha_k} \bar{\mathbf{b}}_k^{\dag} \mathbf{T} \mathbf{P} \mathbf{d}_{k} \right] \bigg\}. & \label{eqn_opt_prob_analog_1}
\end{flalign}
Using the expressions for $\mathbf{T}$ and $\mathbf{d}_{k}$ from \eqref{eqn_tx_signal}, it can be shown that \eqref{eqn_opt_prob_analog_1} can be decoupled for each TTD $n$. For any TTD $n \in \{1,...,N\}$ let us define the set of antenna elements (and phase-shifters) connected to it be represented by $\mathcal{M}_n$. Then \eqref{eqn_opt_prob_analog_1} can be decoupled into $N$ optimization problems with the $n$-th one being:
\begin{flalign}
& \ \ {\tau}_n^{\circ}(\angle\boldsymbol{\alpha}), \{\phi_m^{\circ}(\angle\boldsymbol{\alpha}) | m \in \mathcal{M}_n \} = \argmax_{\tau_n, \{\phi_m | m \in \mathcal{M}_n \} } \bigg\{ & \nonumber \\
& \quad \sum_{k \in \mathcal{K}} \sum_{m \in \mathcal{M}_n} \omega_k {\mathrm{Re}} \left[ e^{{\mathrm{j}} \angle\alpha_k} {[\bar{\mathbf{b}}_k]}^{*}_m e^{{\mathrm{j}} \phi_{m}} e^{-{\mathrm{j}} 2 \pi f_k \tau_{n}} \right] \bigg\}. & \label{eqn_opt_prob_analog_2}
\end{flalign}
The optimal solution to \eqref{eqn_opt_prob_analog_2} for any $n$ and $m \in \mathcal{M}_n$ can be easily shown to be:
\begin{subequations}
\begin{flalign}
& \ \ {\tau}_n^{\circ}(\angle\boldsymbol{\alpha}) = \argmax_{\tau_n} \Bigg\{ & \nonumber \\
& \qquad \quad \sum_{m \in \mathcal{M}_n} \left| \sum_{k \in \mathcal{K}} \omega_k e^{{\mathrm{j}} \angle\alpha_k} {[\bar{\mathbf{b}}_k]}^{*}_m e^{-{\mathrm{j}} 2 \pi f_k \tau_{n}} \right| \Bigg\} \!\!\!\!\!\!\!\!\!\! & \label{eqn_opt_TTD} \\
& \ \ \phi_m^{\circ}(\angle\boldsymbol{\alpha}) = \angle \left[  \sum_{k \in \mathcal{K}} \omega_k e^{-{\mathrm{j}} \angle\alpha_k} {[\bar{\mathbf{b}}_k]}_m e^{{\mathrm{j}} 2 \pi f_k \tau_{n}^{\circ}} \right]. & \label{eqn_opt_phi}
\end{flalign}
\end{subequations}
Note that although \eqref{eqn_opt_TTD} is not in closed form, it can be solved using a line-search over all values of $\tau_n$ within the candidate range: $-\frac{\kappa}{2W} \leq \tau_n \leq \frac{\kappa}{2W}$. 
Alternatively. we have the following lemma:
\begin{lemma} \label{lemma_wLS}
An approximate closed-form solution to \eqref{eqn_opt_TTD} can be obtained, by solving as a first step, the following unconstrained weighted least squares (wLS) problem:
\begin{align}
{\tau}_n^{\circ}(\angle\boldsymbol{\alpha}) = \argmin_{\tau_n} \min_{\{\phi_m | m \in \mathcal{M}_n\}} \bigg\{ \sum_{m \in \mathcal{M}_n} \sum_{k \in \mathcal{K}} \omega_k \nonumber \\
\big| {[\bar{\mathbf{b}}_k]}_m \big| {\left[ 2 \pi f_k \tau_{n} - \phi_m + \mathscr{U}\Big(\angle {[\bar{\mathbf{b}}_k]}_{m} - \angle \alpha_k \Big)\right]}^2 \bigg\}, \label{eqn_lemma_wLS}
\end{align}
where $\mathscr{U}(\cdot)$ is the phase unwrapping function that for each $k$ adds integer shifts of $2 \pi$ to the argument to ensure that phase-difference between adjacent sub-carriers satisfies: 
$$\Big|\mathscr{U} \big(\angle {[\bar{\mathbf{b}}_k]}_{m} - \angle \alpha_k \big) - \mathscr{U} \big(\angle {[\bar{\mathbf{b}}_{k-1}]}_{m} - \angle \alpha_{k-1} \big) \Big| \leq \pi.$$ 
And then as a second step, performing clipping operations:
\begin{subequations} \label{eqn_wLS_bounds}
\begin{align}
\tau_n &= {\mathrm{mod}} \Big\{ \tau_n + \frac{K}{2W} , \frac{K}{W} \Big\} - \frac{K}{2W}, \\
\tau_n &= \max \Big\{ \min \Big\{ \tau_n, \frac{\kappa}{2W} \Big\}, -\frac{\kappa}{2W} \Big\}.
\end{align}
\end{subequations}
\end{lemma}
\begin{proof}
See Appendix \ref{appdix1}.
\end{proof}

Next, for a given $\boldsymbol{\phi}, \boldsymbol{\tau}$, the optimal solution to $\angle \boldsymbol{\alpha}$ in \eqref{eqn_opt_prob_analog} can be found as:
\begin{align}
\angle \alpha_k = \angle \left[ \sum_{n=1}^{N} \sum_{m \in \mathcal{M}_n} {[\bar{\mathbf{b}}_k]}_m e^{-{\mathrm{j}} \phi_{m}} e^{{\mathrm{j}} 2 \pi f_k \tau_{n}} \right] \label{eqn_opt_psi}.
\end{align}
Now using \eqref{eqn_opt_TTD}, \eqref{eqn_opt_phi} and \eqref{eqn_opt_psi} recursively, we can design the JPTA analog beamformer. This iterative approach is summarized in Algorithm \ref{Algo1}, where we also use an intermediate step of adding a constant offset to all the TTD values to push them to the middle of the `feasible' range. It can be verified that removing a common offset $\bar{\tau}$ from all the TTDs doesn't affect the optimality of the solution $\{\tau_n | 1 \leq n \leq N\}$ since the change in JPTA beamformer phase can be countered by adding $\angle\alpha_k \leftarrow \angle\alpha_k - 2 \pi f_k \bar{\tau}$. 
\begin{algorithm}
\caption{Iterative optimization of $\boldsymbol{\tau}, \boldsymbol{\phi}, \angle \boldsymbol{\alpha}$}
\label{Algo1}
\begin{algorithmic} 
\STATE Given: ${\mathbf{b}}_k$ for $k \in \mathcal{K}$
\STATE Compute $|\alpha_k| = \|{\mathbf{b}}_k\|$ for each $k \in \mathcal{K}$.
\STATE Compute $\bar{\mathbf{b}}_k \triangleq {\mathbf{b}}_k/\|{\mathbf{b}}_k\|$ for each $k \in \mathcal{K}$.
\STATE Initialize $\angle \alpha_k = 0$ for each $k \in \mathcal{K}$.
\FOR{$i=1:1:\text{max-iter}$}
\STATE \% \textit{Optimize $\boldsymbol{\tau}, \boldsymbol{\phi}$ for given $\angle\boldsymbol{\alpha}$}
\FOR{$n=1:1:N$}
\STATE Compute $\tau_n$ using either \eqref{eqn_opt_TTD} or Lemma \ref{lemma_wLS}.
\STATE Compute $\phi_m$ using \eqref{eqn_opt_phi} for each $m \in \mathcal{M}_n$.
\ENDFOR
\STATE \% \textit{Push $\boldsymbol{\tau}$ to center of the allowed range.}
\STATE ${\tau}_{\mathrm{min}} = \min_{n} \{\tau_{n} \}$
\STATE ${\tau}_{\mathrm{max}} = \max_{n} \{\tau_{n} \}$
\STATE $\bar{\tau} = \max \left\{ \min \left\{ \sum_{n} \frac{\tau_{n}}{N}, \frac{\kappa}{2W} + {\tau}_{\mathrm{min}} \right\}, {\tau}_{\mathrm{max}} - \frac{\kappa}{2W}\right\}$
\STATE Set $\tau_n \leftarrow \tau_n - \bar{\tau}$ for each $n$.
\STATE Set $\angle \alpha_k \leftarrow \angle \alpha_k - 2 \pi f_k \bar{\tau}$ for each $k \in \mathcal{K}$
%
\STATE \% \textit{Optimize $\angle\boldsymbol{\alpha}$ for given $\boldsymbol{\tau}, \boldsymbol{\phi}$}
\FOR{$k \in \mathcal{K}$}
\STATE Compute $\angle\alpha_k$ using \eqref{eqn_opt_psi}.
\ENDFOR
\ENDFOR
\IF{$\tau_n \geq 0$ required (see Section \ref{sec_prob_formulate})} 
\STATE $\tau_{\rm min} = \min_{1\leq n \leq N}\{ \tau_n\}$.
\STATE $\tau_n \leftarrow \tau_n - \tau_{\rm min}$ for each $n$.
\STATE $\angle \alpha_k \leftarrow \angle \alpha_k - 2 \pi f_k \tau_{\rm min}$ for each $k \in \mathcal{K}$.
\ENDIF
\STATE Return values of $\{\tau_n | 1 \leq n \leq N\}$, $\{\phi_m | 1 \leq m \leq M\}$ and $\{\alpha_k | k \in \mathcal{K}\}$.
\end{algorithmic}
\end{algorithm}
Note that in this work we have considered a continuous TTD range of $0 \leq \tau_n \leq \kappa / W$. If we consider a discrete set of feasible values for TTDs, we have the following remark:
\begin{remark}
If the TTD feasible range is a discrete set, \eqref{eqn_opt_TTD} can be solved by a line-search over the discrete set. Due to the convexity of the objective, \eqref{eqn_lemma_wLS} can be solved by first obtaining the TTD solution considering a continuous TTD range, and then rounding it to the nearest discrete value.
\end{remark}

\section{Baseline algorithms} \label{sec_baseline_algos}
In this section, we propose some simple heuristic designs of JPTA parameters to realize beam behaviors 1 and 2. While providing closed-form (unlike Algorithm \ref{Algo1}) solutions for these beam behaviors, they also serve as baselines to evaluate the efficacy of Algorithm \ref{Algo1}. However, it should be emphasized that unlike these heuristics, Algorithm \ref{Algo1} can approximate any desired beam behavior $\{\mathbf{b}_k | k \in \mathcal{K}\}$ as shall also be shown in Section \ref{sec_future_dir}. 

\subsubsection{Behavior 1} For beam behavior 1, we have: 
$${[\bar{\mathbf{b}}_k]}_m = e^{{\mathrm{j}} \pi m \sin \left(\theta_1 + \frac{k}{K} \Delta \theta \right) f_k/f_0 } /\sqrt{M},$$
which has an `almost' linear phase variation as a function of $k$ with a slope of: $m \pi [\sin(\theta_0 + \frac{\Delta \theta}{2})f_{\mathrm{max}} - \sin(\theta_0-\frac{\Delta\theta}{2}) f_{\mathrm{min}}]/(W f_0)$, where $f_{\mathrm{max}} \triangleq f_{\lfloor (K-1)/2 \rfloor}$, $f_{\mathrm{min}} \triangleq f_{\lfloor (1-K)/2 \rfloor}$ and $W =  f_{\mathrm{max}} - f_{\mathrm{min}}$ is the system bandwidth. This linear phase-variation can be realized using a TTD. Since all antennas in set $\mathcal{M}_n$ share the same TTD $\tau_n$, we take the mean of the phase-variation slopes for these antennas to compute the value of $\tau_n$. The phase-shifts $\phi_m$ are then set to ensure an exact match to ${[\bar{\mathbf{b}}_k]}_m$ at the center sub-carrier $k=0$. This algorithm is summarized in Algorithm \ref{Algo2}.\footnote{This algorithm can be interpreted as a generalization of the heuristic proposed in \cite{Design_Boljanovic2020, Dynamic_Yan2021, Fast_Boljanovic2021}.} 
\begin{algorithm} [t]
\caption{Heuristic solution for behavior 1}
\label{Algo2}
\begin{algorithmic} 
\STATE Given $\theta_0, \Delta \theta$.
\STATE Compute $|\alpha_k| = \sqrt{P_{\textrm{sum}}/K}$ for each $k \in \mathcal{K}$.
\STATE Define $f_{\mathrm{max}} \triangleq f_{\lfloor (K-1)/2 \rfloor}$, $f_{\mathrm{min}} \triangleq f_{\lfloor (1-K)/2 \rfloor}$
\FOR{$n=1:1:N$}
\STATE $\tau_n = \sum_{m \in \mathcal{M}_n} \frac{m[\sin(\theta_0 - \frac{\Delta \theta}{2}) f_{\mathrm{min}} - \sin(\theta_0 + \frac{\Delta \theta}{2})f_{\mathrm{max}}]}{2 W f_0 |\mathcal{M}_n|}$.
\ENDFOR
\STATE Set $\bar{\tau} = \sum_{n=1}^{N} \tau_{n} / {N}$.
\STATE Set $\tau_n \leftarrow \tau_n - \bar{\tau}$ for each $n$.
\STATE Set $\tau_n \leftarrow \min \big\{ \max\{\tau_n, -\frac{\kappa}{2W}\}, \frac{\kappa}{2W} \big\}$ for each $n$
\FOR{$n=1:1:N$}
\FOR{$m \in \mathcal{M}_n$}
\STATE $\phi_m = \pi (m-1) \sin(\theta_0) + 2 \pi f_{0} \tau_n$.
\STATE $\phi_m \leftarrow {\mathrm{mod}} \left\{ \phi_m + \pi, 2 \pi \right\} - \pi$.
\ENDFOR
\ENDFOR
\FOR{$k \in \mathcal{K}$}
\STATE Compute $\angle\alpha_k$ using \eqref{eqn_opt_psi}.
\ENDFOR
\IF{$\tau_n \geq 0$ required (see Section \ref{sec_prob_formulate})} 
\STATE $\tau_{\rm min} = \min_{1\leq n \leq N}\{ \tau_n\}$.
\STATE $\tau_n \leftarrow \tau_n - \tau_{\rm min}$ for each $n$.
\STATE $\angle \alpha_k \leftarrow \angle \alpha_k - 2 \pi f_k \tau_{\rm min}$ for each $k \in \mathcal{K}$.
\ENDIF
\STATE Return values of $\{\tau_n | 1 \leq n \leq N\}$, $\{\phi_m | 1 \leq m \leq M\}$ and $\{\alpha_k | k \in \mathcal{K}\}$.
\end{algorithmic}
\end{algorithm}

\subsubsection{Behavior 2} For beam behavior 2, note that: 
\begin{align}
{[\bar{\mathbf{b}}_k]}_m = \left\{ \begin{array}{ll} 
e^{{\mathrm{j}} \pi m \sin(\theta_1)f_k/f_0} / \sqrt{M} & \text{for $k < 0$} \\
e^{{\mathrm{j}} \pi m \sin(\theta_2)f_k/f_0} / \sqrt{M} & \text{for $k \geq 0$}
\end{array} \right. , \nonumber
\end{align}
i.e., a `step-function' like phase-variation is required with $k$ for each antenna $m$. Since such a sharp step-function can't achieved using TTDs, we design the TTDs to realize a linear-phase approximation to this step function, that also passes through the sum of the two array responses: 
$$[\widetilde{\mathbf{b}}]_{m} =  [e^{{\mathrm{j}} \pi m \sin(\theta_2)} + e^{{\mathrm{j}} \pi m \sin(\theta_2)}] \Big/ \sqrt{2 M}.$$
The slope of the linear-approximation is then given by: $3 \angle \big[ {[\widetilde{\mathbf{b}}]}^{*}_m e^{{\mathrm{j}} \pi m \sin(\theta_2)f_k/f_0} \big]/W$. Since all antennas in set $\mathcal{M}_n$ share the same TTD $\tau_n$, we take the mean of the phase-variation slopes for these antennas to compute the value of $\tau_n$. The phase-shifts $\phi_m$ are set to ensure that at $k=0$ the beam shape is aligned with $\widetilde{\mathbf{b}}$. This algorithm is summarized in Algorithm \ref{Algo3}.
\begin{algorithm} [t]
\caption{Heuristic solution for behavior 2}
\label{Algo3}
\begin{algorithmic} 
\STATE Given $\theta_1, \theta_2$.
\STATE Compute $|\alpha_k| = \sqrt{P_{\textrm{sum}}/K}$ for each $k \in \mathcal{K}$.
\STATE ${[\widetilde{\mathbf{b}}]}_{m} =  [e^{{\mathrm{j}} \pi m \sin(\theta_2)} + e^{{\mathrm{j}} \pi m \sin(\theta_2)}] \big/\sqrt{2M}.$
\FOR{$n=1:1:N$}
\STATE $\tau_n = \frac{-3}{2\pi W} \angle \left( \sum_{m \in \mathcal{M}_n} {[\widetilde{\mathbf{b}}]}^{*}_m e^{{\mathrm{j}} \pi m \sin(\theta_2)f_{\lfloor K/3 \rfloor}/f_{0}} \right) $.
\ENDFOR
\STATE Set $\bar{\tau} = \sum_{n=1}^{N} \tau_{n} / {N}$.
\STATE Set $\tau_n \leftarrow \tau_n - \bar{\tau}$ for each $n$.
\STATE Set $\tau_n \leftarrow \min \big\{ \max\{\tau_n, -\frac{\kappa}{2W}\}, \frac{\kappa}{2W} \big\}$ for each $n$.
\FOR{$n=1:1:N$}
\FOR{$m \in \mathcal{M}_n$}
\STATE $\phi_m = \angle{[\widetilde{\mathbf{b}}]}_m + 2 \pi f_{0} \tau_n$.
\STATE $\phi_m \leftarrow {\mathrm{mod}} \left\{ \phi_m + \pi, 2 \pi \right\} - \pi$.
\ENDFOR
\ENDFOR
\FOR{$k \in \mathcal{K}$}
\STATE Compute $\angle\alpha_k$ using \eqref{eqn_opt_psi}.
\ENDFOR
\IF{$\tau_n \geq 0$ required (see Section \ref{sec_prob_formulate})} 
\STATE $\tau_{\rm min} = \min_{1\leq n \leq N}\{ \tau_n\}$.
\STATE $\tau_n \leftarrow \tau_n - \tau_{\rm min}$ for each $n$.
\STATE $\angle \alpha_k \leftarrow \angle \alpha_k - 2 \pi f_k \tau_{\rm min}$ for each $k \in \mathcal{K}$.
\ENDIF
\STATE Return values of $\{\tau_n | 1 \leq n \leq N\}$, $\{\phi_m | 1 \leq m \leq M\}$ and $\{\alpha_k | k \in \mathcal{K}\}$.
\end{algorithmic}
\end{algorithm}
Based on the heuristic algorithms we also have the following remark:
\begin{remark} \label{rem_maxTTD}
The maximum TTD delays required to realize beam-behaviors 1 and 2 using the heuristic solutions in Algorithms \ref{Algo2} and \ref{Algo3} (without the need for clipping) are: $M |\sin(\Delta \theta/2)|/W$ and $3/W$, respectively. 
\end{remark}
In other words, the required TTD range scales linearly with number of antennas $M$ for behavior 1 but doesn't scale with $M$ for behavior 2.

\section{Simulation results} \label{sec_simulation_results}
For simulations we consider the downlink of a THz system, operating at a carrier frequency of $100$ GHz and OFDM modulation with $K = 2048$ sub-carriers indexed as $\mathcal{K} = \{-1024,...,1023\}$. The system bandwidth is set to $W=10$ GHz with the sub-carrier frequencies set as: $f_k = 100 + 10k/K$ GHz.\footnote{The impact of cyclic prefix on bandwidth is ignored here for convenience.} The BS has half-wavelength-spaced (at $100$GHz) uniform linear array with $N=64$ antenna elements. The BS has a JPTA architecture with one RF chain, $N \leq M$ TTDs and $M$ phase-shifters, and the mapping matrix $\mathbf{P}$ is such that we have $\mathcal{M}_n = \{m \in \mathbb{N} | (n-1)M/N < m \leq nM/N \}$. This connection structure is illustrated in Fig.~\ref{Fig_connected_structures}. Note that this allows adjacent antenna elements to share the same TTD when $N < M$, which makes hardware routing easier. Several other architectures with varying complexity that are well suited for different beam-behaviors, a few of which are discussed later in Section \ref{sec_future_dir}. 
For emulating behavior 1 and 2, we use the parameters $\theta_0 = \pi/6$, $\Delta \theta = \pi/4$, $\theta_1 = -\pi/4$, $\theta_2 = \pi/6$. For all the results in this section, we set the sub-carrier weights as: $\omega_k=1$ for all $k \in \mathcal{K}$. 
%
\begin{figure}[!htb]
\centering
\includegraphics[width= 0.45\textwidth]{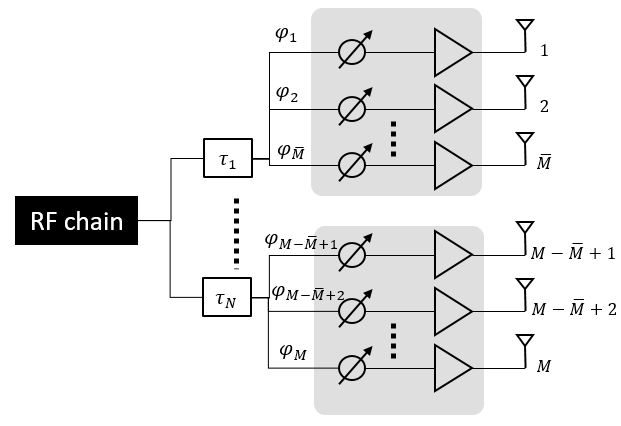}
\caption{An illustration of the proposed architecture for the mapping matrix $\mathbf{P}$. In these figures, each TTD is connected to $\bar{M} = M/N$ antennas.}
\label{Fig_connected_structures}
\end{figure}

First we consider the optimistic scenario with one TTD per antenna $N=M$ and a large adjustable TTD range $\kappa = M$ (this implies an TTD adjustable range $[0,6.4]$ ns). For this case, the achieved beamforming gain with JPTA using Algorithm \ref{Algo1} are compared to the ideal beamformers $\{\bar{\mathbf{b}}_k | k \in \mathcal{K}\}$ in Fig.~\ref{Fig_ant_gain_JPTA}. As can be seen from the results, JPTA with abundant TTDs and a large adjustable range can accurately achieve the desired beam behavior 1. For behavior 2, although the desired sharp transition cannot be achieved, we indeed observe that the main beam lobe switches from one angle to another gradually, achieving the desired behavior. 
\begin{figure}[!htb]
\centering
\subfloat[Ideal, Behavior 1]{\includegraphics[width= 0.24\textwidth]{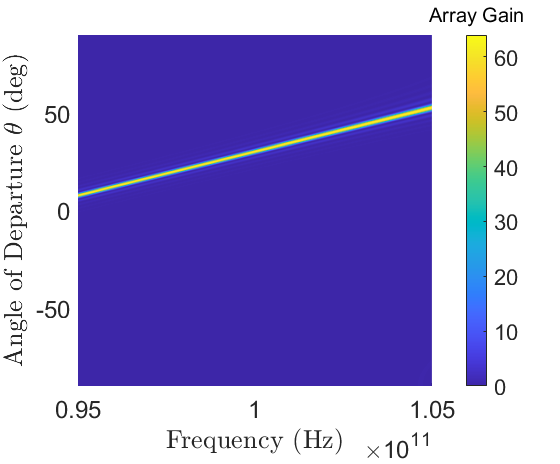} \label{Fig_behav1_ideal}}
\subfloat[JPTA, Behavior 1]{\includegraphics[width= 0.24\textwidth]{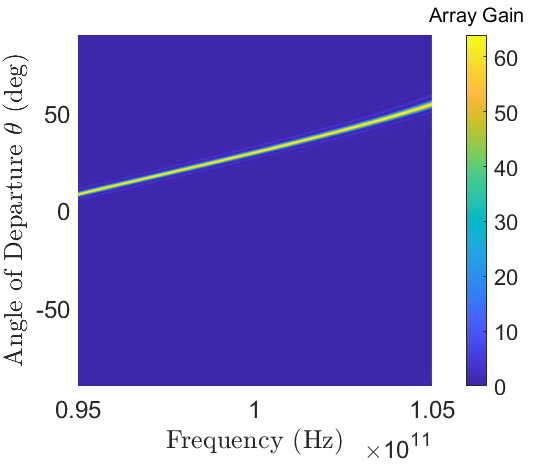} \label{Fig_behav1_algo1}} \\
\subfloat[Ideal, Behavior 2]{\includegraphics[width= 0.24\textwidth]{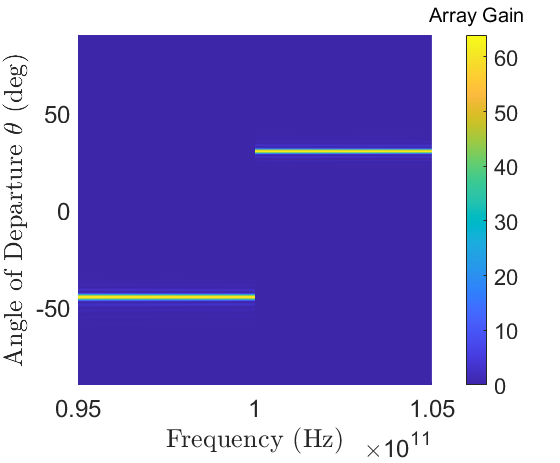} \label{Fig_behav2_ideal}}
\subfloat[JPTA, Behavior 2]{\includegraphics[width= 0.24\textwidth]{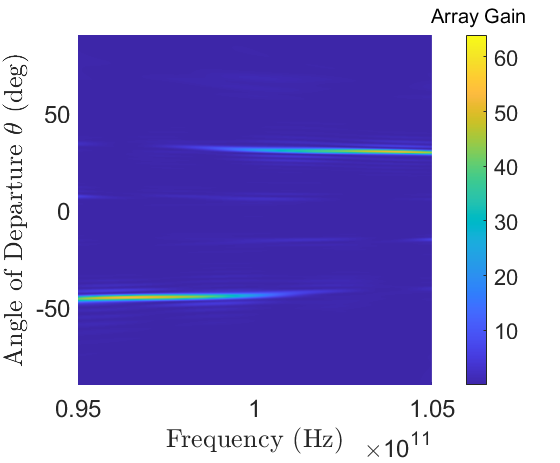} \label{Fig_behav2_algo1}}
\caption{An illustration of the array gain: ${|\mathbf{a}_k(\theta)^{\dag} \mathbf{T} \mathbf{P} \mathbf{d}_{k}|}^2$, achievable with JPTA beamforming with Algorithm \ref{Algo1} (with \eqref{eqn_opt_TTD}) to replicate beam behaviors 1 ($\theta_0 = \pi/6$, $\Delta \theta = \pi/4$) and behavior 2 ($\theta_1 = -\pi/4$, $\theta_2 = \pi/6$). Here we use $\text{max-iter} = 10$ for Algorithm \ref{Algo1}.}
\label{Fig_ant_gain_JPTA}
\end{figure}
Since \eqref{eqn_opt_prob_digital} can be solved exactly, in the rest of the section we shall use the objective in \eqref{eqn_opt_prob_analog}, i.e.,
\begin{align} \label{eqn_f_obj_defn}
\mathscr{F}_{\mathrm{obj}} = \sum_{k \in \mathcal{K}} \omega_k \left|  {\bar{\mathbf{b}}}_k^{\dag} \mathbf{T} \mathbf{P} \mathbf{d}_{k} \right| \Big/ \sum_{\bar{k} \in \mathcal{K}} \omega_{\bar{k}} ,
\end{align}
to quantify the `goodness of fit' between JPTA beamforming and the desired behavior $\boldsymbol{\mathcal{B}}$. 

In Fig.~\ref{Fig_JPTA_impact_of_N}, we study the impact of reducing the number of TTDs $N$ on $\mathscr{F}_{\mathrm{obj}}$ for both beam behavior 1 and 2. Here we also include a comparison to the heuristic designs in Algorithm \ref{Algo2} and \ref{Algo3}. As can be observed from the results, Algorithm \ref{Algo1} can outperform  the heuristic methods for all values of $N$. In addition, we observe that the use of Lemma \ref{lemma_wLS} in Algorithm \ref{Algo1} reduces the computational complexity significantly, albeit with a negligibly small drop in $\mathscr{F}_{\mathrm{obj}}$. It can also be observed that, in its proposed architecture, JPTA can replicate behavior 1 more accurately than behavior 2. However, behavior 2 is more robust to reducing the number of TTDs $N$ in comparison to behavior 1.
\begin{figure}[!htb]
\centering
\subfloat[Behavior 1 ($\theta_0 = \pi/6$, $\Delta \theta = \pi/4$)]{\includegraphics[width= 0.45\textwidth]{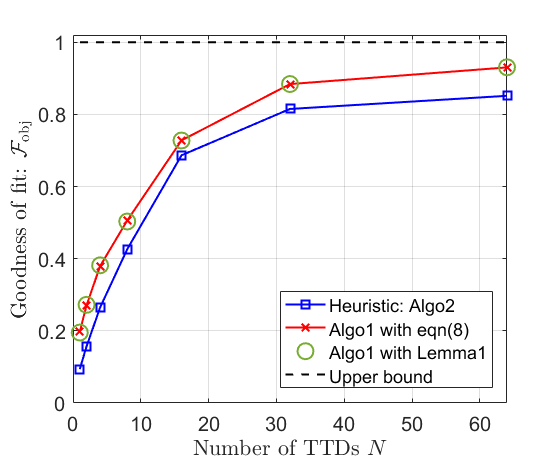} \label{Fig_varyN_behav1}} \\ 
\subfloat[Behavior 2 ($\theta_1 = -\pi/4$, $\theta_2 = \pi/6$)]{\includegraphics[width= 0.45\textwidth]{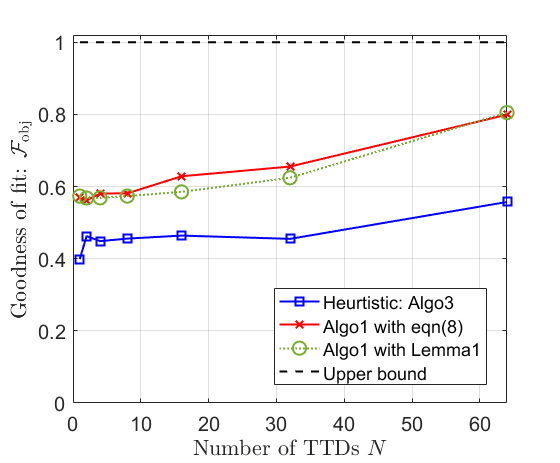} \label{Fig_varyN_behav2}}
\caption{An illustration of the objective $\mathscr{F}_{\mathrm{obj}}$ achievable with JPTA under different design algorithms, with varying number of TTDs used. Here we use $\text{max-iter} = 10$ for Algorithm \ref{Algo1}.}
\label{Fig_JPTA_impact_of_N}
\end{figure}

Next in Fig.~\ref{Fig_JPTA_impact_of_kappa}, we study the impact of varying the adjustable range of the TTDs $\kappa$ on beam behavior 1 and 2. As can be observed from the results, Algorithm \ref{Algo1} can outperform  the heuristic methods for all values of $\kappa$. In addition, we observe that the use of Lemma \ref{lemma_wLS} in Algorithm \ref{Algo1} reduces the computational complexity significantly, albeit with a negligibly small drop in $\mathscr{F}_{\mathrm{obj}}$. From the results, we also observe that $\mathscr{F}_{\mathrm{obj}}$ saturates beyond a value of $\kappa = M \sin(\Delta \theta)$ for behavior 1 and $\kappa = 2$ for behavior 2, which is in alignment with Remark \ref{rem_maxTTD}. 
\begin{figure}[!htb]
\centering
\subfloat[Behavior 1 ($\theta_0 = \pi/6$, $\Delta \theta = \pi/4$)]{\includegraphics[width= 0.45\textwidth]{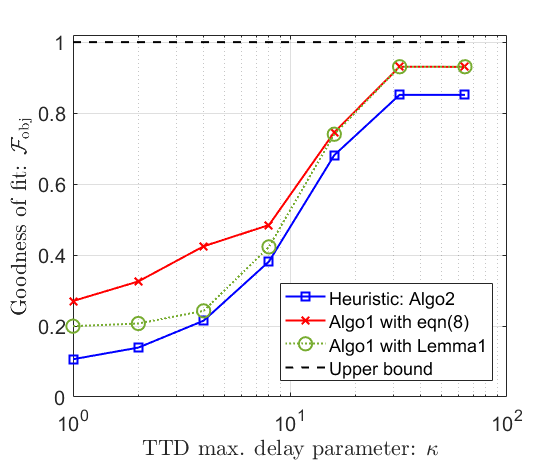} \label{Fig_varykappa_behav1}} \\ 
\subfloat[Behavior 2 ($\theta_1 = -\pi/4$, $\theta_2 = \pi/6$)]{\includegraphics[width= 0.45\textwidth]{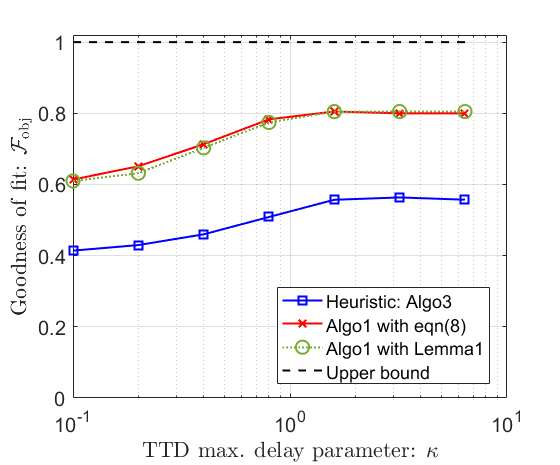} \label{Fig_varykappa_behav2}}
\caption{An illustration of the objective $\mathscr{F}_{\mathrm{obj}}$ achievable with JPTA under different design algorithms, with varying TTD range $\tau_n \in [0,\kappa/W]$. Here we use $\text{max-iter} = 10$ for Algorithm \ref{Algo1}.}
\label{Fig_JPTA_impact_of_kappa}
\end{figure}

Finally, we study the convergence speed of Algorithm \ref{Algo1} as a function of the number of iterations $\text{max-iter}$. Let us define $\mathscr{F}_{\mathrm{obj}}(i)$ to be the value in \eqref{eqn_f_obj_defn} achieved with $\text{max-iter}=i$ in Algorithm \ref{Algo1}. In Fig.~\ref{Fig_convergence_algo1}, the we plot the mean, $10$th percentile and $90$th percentile values of $\mathscr{F}_{\mathrm{obj}}(i)/\mathscr{F}_{\mathrm{obj}}(30)$, where the statistics are computed over many different realizations of $\kappa, N, \theta_0, \Delta \theta$ for behavior 1 and $\kappa, N, \theta_1, \theta_2$ for behavior 2. As can be seen from the results, $\text{max-iter}=10$ is sufficient to ensure Algorithm \ref{Algo1} converges for both beam behaviors, although convergence is faster for beam-behavior 2.\footnote{Although we initialize $\angle\alpha_k=0$ in Algorithm \ref{Algo1} to get these results, similar convergence was also observed for other random initializations.}
\begin{figure}[!htb]
\centering
\subfloat[Behavior 1]{\includegraphics[width= 0.45\textwidth]{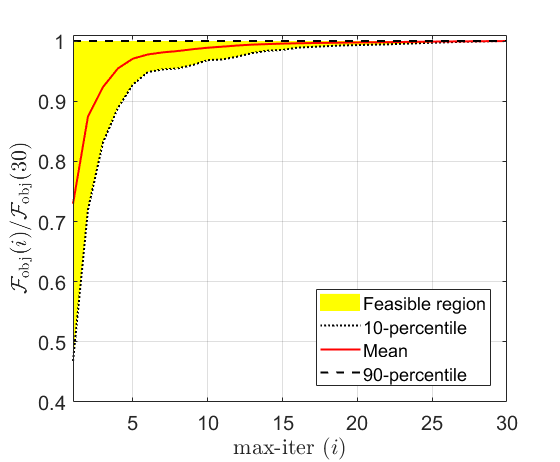} \label{Fig_convergence_algo1_behav1}} \\ 
\subfloat[Behavior 2]{\includegraphics[width= 0.45\textwidth]{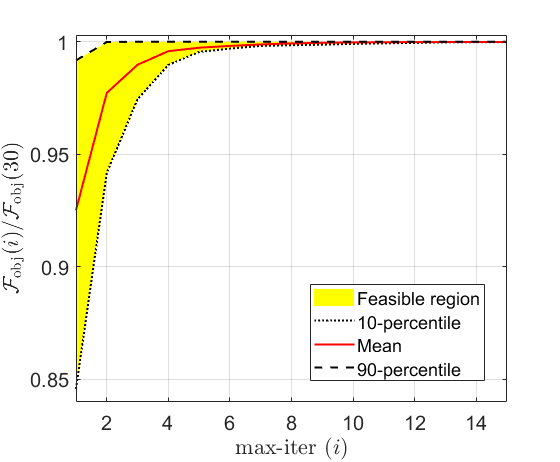} \label{Fig_convergence_algo1_behav2}}
\caption{An illustration of convergence speed of Algorithm \ref{Algo1} versus $\text{max-iter}$. Results show statistics of $\mathscr{F}_{\mathrm{obj}}(i)/\mathscr{F}_{\mathrm{obj}}(30)$ computed over many different random realizations of beam parameters and JPTA components.}
\label{Fig_convergence_algo1}
\end{figure}

\section{JPTA versus conventional HBF} \label{sec_compare_hybrid_BF}

\begin{figure}[!htb]
\centering
\subfloat[Fully-connected]{\includegraphics[width= 0.24\textwidth]{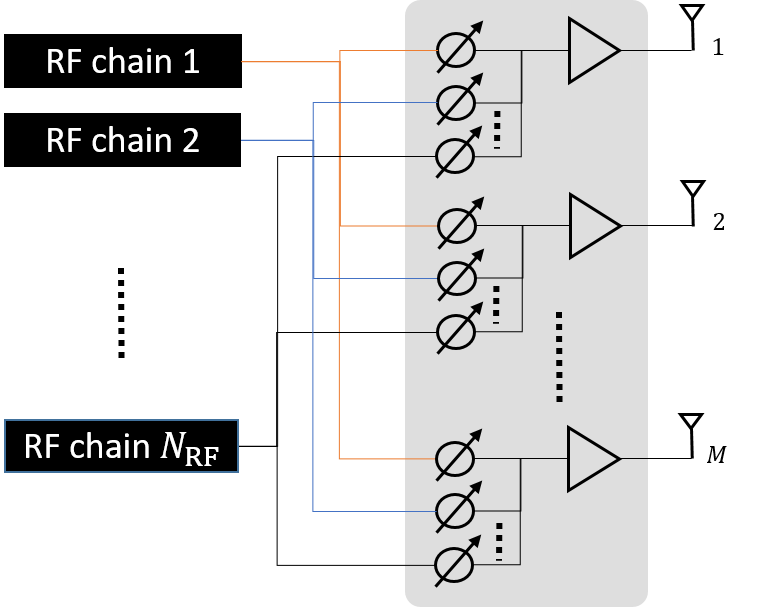} \label{Fig_HBF_FC_architecture}} 
\subfloat[Partially-connected]{\includegraphics[width= 0.24\textwidth]{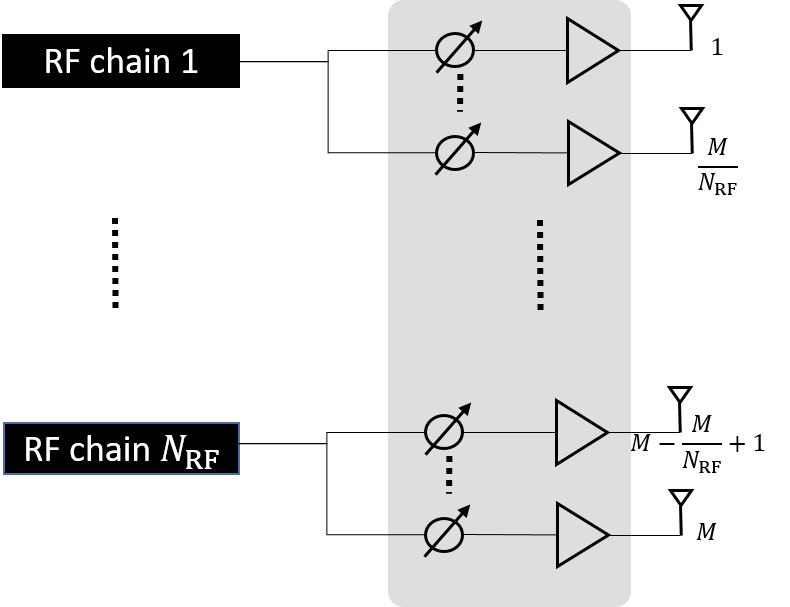} \label{Fig_HBF_PC_architecture}}
\caption{An illustration of two conventional HBF architectures.}
\label{Fig_HBF_architecture}
\end{figure}

In conventional HBF, while the analog hardware components have a frequency-flat response, frequency-dependent beamforming can still be realized in the digital domain by using multiple RF chains. Therefore, for comparison to JPTA, in this section we study how many RF chains are required to emulate the beam behaviors in Section \ref{sec_des_beam_behave} using conventional HBF structures. 
In particular, we consider two HBF structures shown in \figref{Fig_HBF_architecture}:
\begin{enumerate}
\item Fully-connected (FC) structure: In this architecture, each antenna is connected to each RF chain via a dedicated phase-shifter, thus requiring $N_{\mathrm{RF}}$ RF chains and $N_{\mathrm{RF}} M$ phase-shifters.
\item Partially-connected (PC) structure: In this structure, the whole antenna array is equally divided into $N_{\mathrm{RF}}$ sub-arrays, with antenna elements $\widetilde{\mathcal{M}}_n = \big\{m \in \mathbb{N} \big| (n-1)M/N_{\mathrm{RF}} < m \leq nM/N_{\mathrm{RF}} \big\}$, being connected to the $n$-th RF chain. This structure requires $N_{\mathrm{RF}}$ RF chains and $M$ phase-shifters.
\end{enumerate}
It is worth highlighting that the main benefit of multiple RF chains, i.e., spatial multiplexing, is not exploited when emulating such behaviors 1 and 2. In addition, the hardware cost and power consumption associated with JPTA with 1 RF chain is much lower than that of HBF with multiple RF chains. Finally, JPTA can also potentially be used in combination with multiple RF chains. Thus this study may not be a true apples-to-apples comparison between JPTA and conventional HBF.

\subsubsection{Fully-connected structure (HBF-FC)}
To emulate a desired beam behavior $\mathcal{B}$ with HBF-FC, we formulate the beamformer design problem as:
\begin{align}
    \min_{\FRF, \mathbf{f}_{{\mathrm{BB}}, k}} & \sum_{k \in \mathcal{K}} {\left\| \mathbf{b}_{k}  - \FRF \fBBk \right \|}^2_F \nonumber \\
    {\mathrm{s.t.}} \quad & \left|\left[\FRF \right]_{i, j} \right|=1, \forall i, j, \nonumber \\
    & \sum_{k \in \mathcal{K}} {\left\|\FRF \fBBk \right\|}^2 = P_{\mathrm{sum}},
\end{align}
where $\FRF \in \mathbb{C}^{M\times N_{\mathrm{RF}}}$ is the phase-shifter-based analog beamforming matrix which is the same across the subcarriers, and $\mathbf{f}_{{\mathrm{BB}}, k} \in \mathbb{C}^{N_{\mathrm{RF}} \times 1}$ is the digital beamforming vector at the $k$-th subcarrier.
By stacking the desired beamformer in a single matrix $\mathbf{B} \triangleq \left[\mathbf{b}_{\lfloor (1-K)/2 \rfloor}, \cdots, \mathbf{b}_{\lfloor (K-1)/2 \rfloor} \right] \in \mathbb{C}^{ M \times K}$, and the digital beamforming in a single matrix $\mathbf{F}_{\mathrm{BB}} \triangleq \left[\mathbf{f}_{{\mathrm{BB}}, \lfloor (1-K)/2 \rfloor}, \cdots, \mathbf{f}_{{\mathrm{BB}}, \lfloor (K-1)/2 \rfloor} \right] \in \mathbb{C}^{ N_{\mathrm{RF}}\times K}$, the optimization problem can be re-written as:
\begin{align}
    \min_{\FRF, \FBB} & {\left\| \mathbf{B}  - \FRF \FBB \right \|}^2_F \label{eq:HBF} \nonumber \\
    {\mathrm{s.t.}} \quad & \left|\left[\FRF \right]_{i, j} \right|=1, \forall i, j, \nonumber \\
    & {\|\FRF \FBB \|}_F^2 = P_{\mathrm{sum}}.
\end{align}
The optimization problem in \eqref{eq:HBF} has the same form as the one formulated in \cite[Eqn. (15)]{Ayach_TWC14}, with the only difference being that each column of $\mathbf{B}$ stands for one beamformer in the frequency domain. Consequently, the phase-extraction alternating minimization (PE-AltMin) algorithm from \cite{Yu2016} can be adopted to find a local optimal solution\footnote{We also tried the orthogonal matching pursuit (OMP) method proposed in \cite{Ayach_TWC14} and saw similar performance as PE-AltMin for our problem.}. 
%

\subsubsection{Partially-connected structure}
For the HBF-PC structure, the beamformer design problem is similar to \eqref{eq:HBF} with the additional constraints that ${[\mathbf{F}_{\mathrm{RF}}]}_{m,n}=0$ for $m \neq \widetilde{\mathcal{M}}_{n}$. 
We generate the analog and digital beamforming in an alternating minimization manner, 
the details of which are skipped here for brevity. 

For a transmitter with a $M=64$ antenna elements, either HBF-FC or HBF-PC structure, and similar signal parameters as Section \ref{sec_simulation_results}, the achievable array gain when emulating behavior 1 ($\theta_0 = \pi/6$, $\Delta \theta = \pi/4$) and behavior 2 ($\theta_1 = -\pi/4$, $\theta_2 = \pi/6$) are illustrated in Fig.~\ref{Fig_ant_gain_HBF}. By comparing to Fig.~\ref{Fig_ant_gain_JPTA}, we observe that HBF can realize a sharper transition of the beam pointing angle with frequency than JPTA, as also reflected in the $\mathscr{F}_{\mathrm{obj}}$ value in Fig.~\ref{Fig_HNF_impact_of_NRF}.
\begin{figure}[!htb]
\centering
\subfloat[HBF-FC, $N_{\mathrm{RF}}=22$, Behavior 1]{\includegraphics[width= 0.25\textwidth]{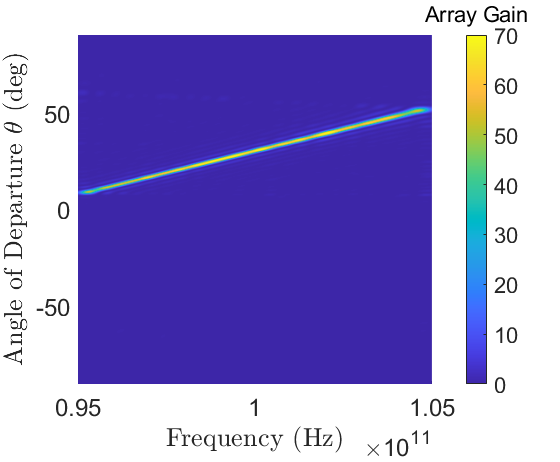} \label{Fig_behav1_FC}} 
\subfloat[HBF-FC, $N_{\mathrm{RF}}=2$, Behavior 2]{\includegraphics[width= 0.25\textwidth]{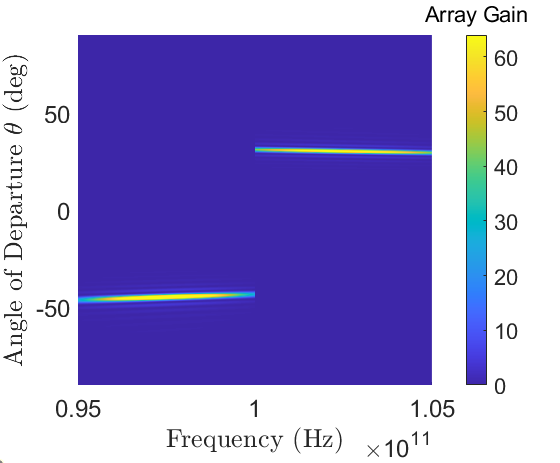} \label{Fig_behav2_FC}} \\
\subfloat[HBF-PC, $N_{\mathrm{RF}}=32$, Behavior 1]{\includegraphics[width= 0.25\textwidth]{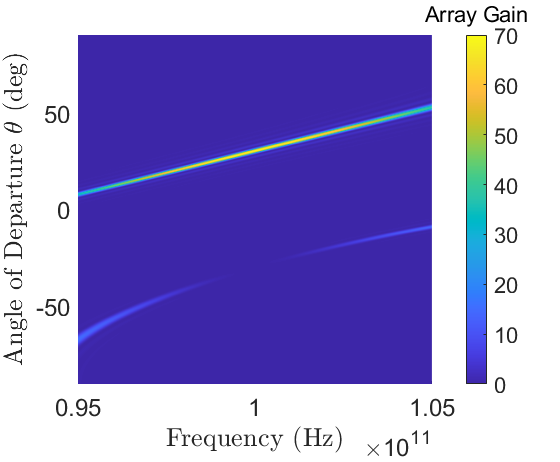} \label{Fig_behav1_PC}} 
\subfloat[HBF-PC, $N_{\mathrm{RF}}=32$, Behavior 2]{\includegraphics[width= 0.25\textwidth]{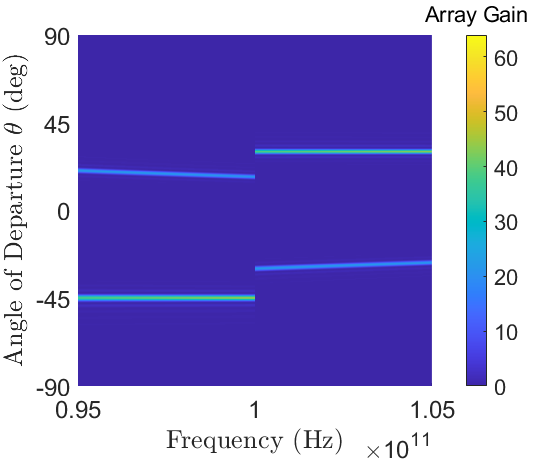} \label{Fig_behav2_PC}} \\
\caption{An illustration of the array gain achievable with HBF-FC and HBF-PC when emulating behavior 1 (with $\theta_0 = \pi/6$, $\Delta \theta = \pi/4$) and behavior 2 (with ($\theta_1 = -\pi/4$, $\theta_2 = \pi/6$)).}
\label{Fig_ant_gain_HBF}
\end{figure} 
Next, the ability of HBF to emulate behavior 1 and 2 as a function of $N_{\mathrm{RF}}$ is studied in Fig.~\ref{Fig_HNF_impact_of_NRF}. Here, for behavior 1, we illustrate two scenarios with $\theta_0 = \pi/6$, $\Delta \theta = \pi/4$ and $\theta_0 = 0$, $\Delta \theta = 2\pi/3$, respectively, while for behavior 2 we consider one scenario with $\theta_1 = -\pi/4$, $\theta_2 = \pi/6$. As a reference, we also plot the performance of JPTA with $1$ RF chain and $N=64$ TTDs. As evident from the results, for these parameters, $22$ RF chains (behavior 1) and $2$ RF chains (behavior 2), respectively are required for HBF-FC to achieve similar performance as JPTA. Similarly $32$ RF chains (behavior 1) and $32$ RF chains (behavior 2), respectively are required for HBF-PC to achieve similar performance as JPTA. With a larger number of RF chains HBF can out-perform JPTA, albeit, at a huge increase in hardware cost and energy consumption. It is worth reiterating that JPTA can be used in conjunction with multiple RF chains to give much better results, which will be explored in future work. Finally, we also observe that the HBF-FC structure can outperform the HBF-PC structure as expected. 
\begin{figure}[!htb]
\centering
\subfloat[Behavior 1]{\includegraphics[width= 0.45\textwidth]{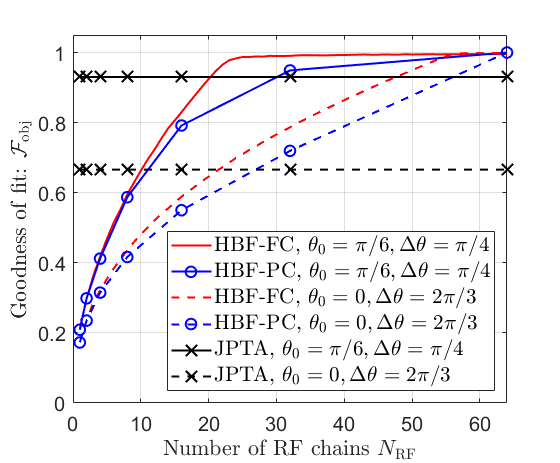} \label{Fig_behav1_HBF_f_obj}} \\ 
\subfloat[Behavior 2]{\includegraphics[width= 0.45\textwidth]{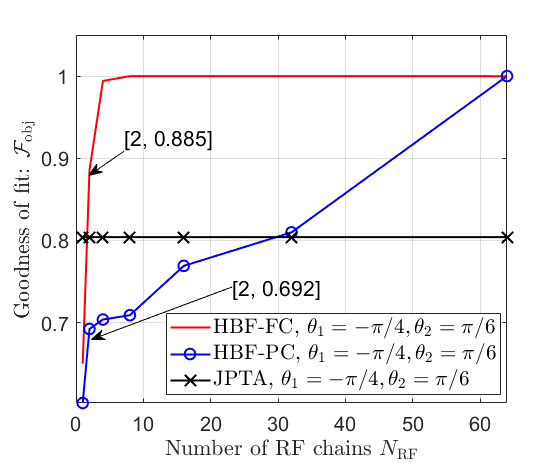} \label{Fig_behav2_HBF_f_obj}}
\caption{An illustration of the objective $\mathscr{F}_{\mathrm{obj}}$ (with $\omega_k=1$) achievable with HBF under different structures, with varying number of RF chains used.}
\label{Fig_HNF_impact_of_NRF}
\end{figure}
Based on these and other experimental observations we also have the following proposition for replicating behavior 1 with HBF:
\begin{proposition} \label{remark_beamsway_FC_PC}
To achieve behavior 1 with HBF-FC, the required minimum number of RF chains is,
\begin{align}
r \triangleq  \Bigg\lceil \frac{M}{2} \bigg| & \sin \left(\theta_0+\frac{\Delta \theta}{2} \right) \frac{f_{\lfloor \frac{K-1}{2} \rfloor}}{f_0} - \nonumber \\
& \sin \left(\theta_0 - \frac{\Delta \theta}{2} \right) \frac{f_{\lfloor \frac{1-K}{2} \rfloor}}{f_0} \bigg| \Bigg\rceil,
\end{align}
while the required minimum number of RF chains for HBF-PC is $2^{\left\lceil \log_2 \left( r \right) \right\rceil}$.
\end{proposition}
\begin{proof}
See Appendix \ref{appdix2}.
\end{proof}
%
%
Note that this proposition implies that the required number of RF chains to emulate these behaviors with conventional HBF will grow linearly with the antenna array size $M$, unlike with JPTA. A comparison of hardware components required for these architectures are summarized below in Table \ref{table1}.
\begin{table}[t]
\centering
\caption{Required number of hardware components to realize both behavior 1 and 2.}
\label{table1}
\begin{tabular}{| c | c | c | c |} 
 \hline
 & \# RF chains & \# Phase-shifters & \# TTDs \\ 
 \hline
 JPTA & $1$ & $M$ & ${\mathrm{O}}(M)$ \\ 
 \hline
 HBF-FC & ${\mathrm{O}}(M)$ &  $M N_{\mathrm{RF}}$ & $0$ \\
 \hline 
 HBF-PC & ${\mathrm{O}}(M)$ & $M$ & $0$ \\
 \hline
\end{tabular}
\end{table}
Note that there are numerous proposed architectures for phase-shifters and TTDs and thus a fair and comprehensive comparison of their hardware costs is beyond the scope of this paper. However, as a ball park, we also include a comparison of the power consumption, insertion loss, and on-chip area of some state-of-the-art implementations of a phase-shifter, a TTD and an RF chain (which includes the mixer, amplifier, and digital-to-analog converter etc.) in Table \ref{table2}. As observed from the results, research towards significantly reducing TTD insertion loss is required to make them as attractive analog-components as phase-shifters. 
\begin{table}[t]
\centering
\caption{Key metrics of transceiver components capable of operating at $f_0=100$ GHz and/or $W=10$ GHz.}
\label{table2}
\begin{tabular}{| p{2.6 cm} | p{0.8 cm} | p{1.5 cm} | p{0.8 cm} |} 
 \hline
Component & Power (mW) & Insertion Loss (dB) & Area ($\mathrm{mm}^2$) \\ 
 \hline
 RF chain \cite{Skrimponis2020} & 125 & -- & -- \\ 
 \hline
 Phase-shifter (Passive) \cite{Li2012_PS, Gu2021} & < 1 & 10 & 0.3 \\
 \hline 
 Phase-shifter (Active) \cite{Testa2020, Pepe2017} & 20 & 5 & 0.1 \\
 \hline
 TTD (max. delay > 50 ps) \cite{Sanggu2013, Hu2015, Dimitrios2021, Song2022} & 0-25 & 10-20 (per 100 ps delay) & 0.3-4 \\
 \hline
\end{tabular}
\end{table}

\section{Future Directions} \label{sec_future_dir}
Although this paper introduces the concept of JPTA and proposes an algorithm for realizing a desired  beam behavior, it has to be emphasized that there is a vast range of topics left to explore. For example, while this work highlights several potential use-cases of frequency-dependent beamforming in Section \ref{sec_des_beam_behave} encompassing coverage extension, capacity enhancement and beam management efficiency, such use-cases require further detailed investigation via system-level analysis. In addition, this work only explores a particular architecture for JPTA where each antenna is connected to a single phase-shifter and TTD. Numerous other multi-path or adaptive (using switches) structures are possible, as exemplified in Fig.~\ref{Fig_other_architectures}, with varying levels of hardware complexity. It is as of yet unclear what type of architecture is most suitable for achieving a certain type of desired beam behavior. Note that the proposed Algorithm \ref{Algo1} may be directly applicable to several of these architectures. 
\begin{figure}[!htb]
\centering
\subfloat[Multi-path fixed architecture]{\includegraphics[width= 0.35\textwidth]{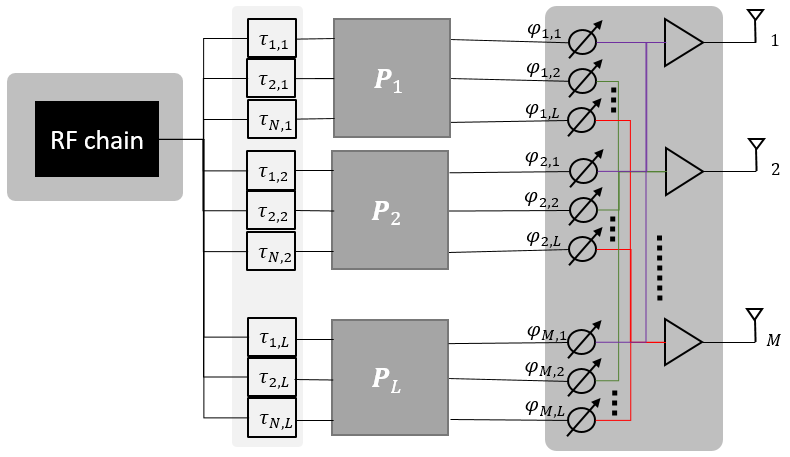} \label{Fig_arch_multipath}} \\
\subfloat[Single-path adaptive architecture]{\includegraphics[width= 0.35\textwidth]{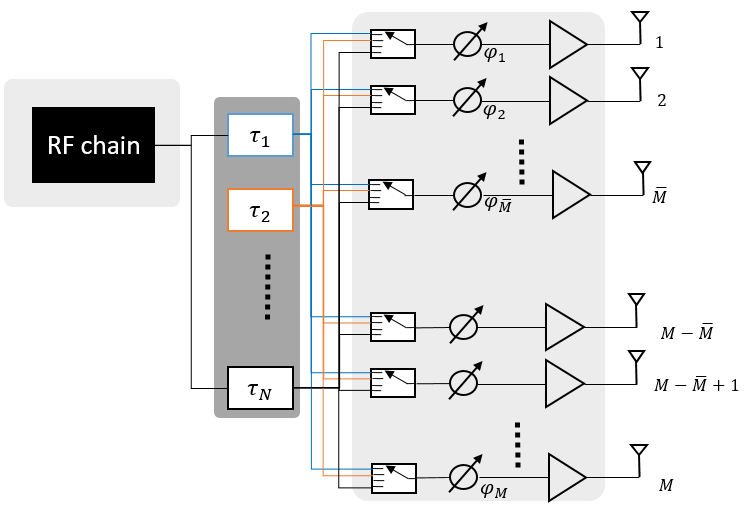} \label{Fig_arch_multipath_switch}}
\caption{An illustration of other potential architectures for JPTA that involve $L$ phase-shifters and TTDs for antenna.}
\label{Fig_other_architectures}
\end{figure}
Similarly, the present paper only considers the case with a single RF chain. In practice, the transceiver may be equipped with multiple RF chains and thus the JPTA design problem with multiple RF chains and its achievable beam behaviors is an interesting problem to explore. The TTD components may also suffer from hardware impairments, as listed in \cite{Fast_Boljanovic2021}, whose impact requires further analysis.

Further, note that the analysis in the paper assumes apriori knowledge of desired beam at each frequency i.e., $\{ \mathbf{b}_k | k \in \mathcal{K}\}$. However in some applications any solution $\{ \mathbf{b}_{\rho(k)} | k \in \mathcal{K}\}$ is also acceptable, where $\rho: \mathcal{K} \rightarrow \mathcal{K}$ is any bijection. In other words, which sub-carriers provide good coverage to each angle in the coverage area may not be important. Additionally, there may be scenarios where for each sub-carrier $k$ a desired set of coverage angles $\{\theta_{k,1},...,\theta_{k,D}\}$ can be provided instead of a desired beamformer $\mathbf{b}_k$. Generalizing Algorithm \ref{Algo1} to cater to such scenarios is also an interesting direction for further exploration. 

Finally, since the concept of frequency-dependent beamforming is new, there is no clear consensus on what frequency-dependent beam variations are desirable. In fact, apart from the behaviors discussed in Section \ref{sec_des_beam_behave}, there can be numerous other examples of desired beam shapes, such as extension of behavior 2 to 3 angles $\{\theta_1, \theta_2, \theta_3\}$, designing a wide coverage-beam with a frequency dependent beam-width to minimize co-channel interference on a sub-band etc. An illustration of the potential of JPTA to form such beams behaviors is depicted in Fig.~\ref{Fig_unorthodox_beam_behav}. It is also worth exploring what are the fundamental limitations on the achievable beam behaviors with JPTA. 
\begin{figure}[!htb]
\centering
\subfloat[Ideal, Behavior 3]{\includegraphics[width= 0.24\textwidth]{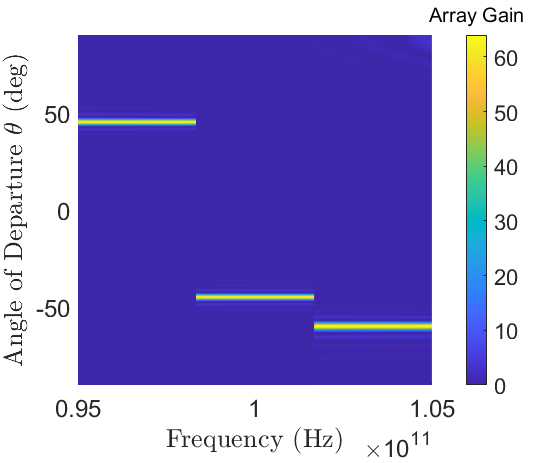} \label{Fig_behav3_ideal}}
\subfloat[JPTA, Behavior 3]{\includegraphics[width= 0.24\textwidth]{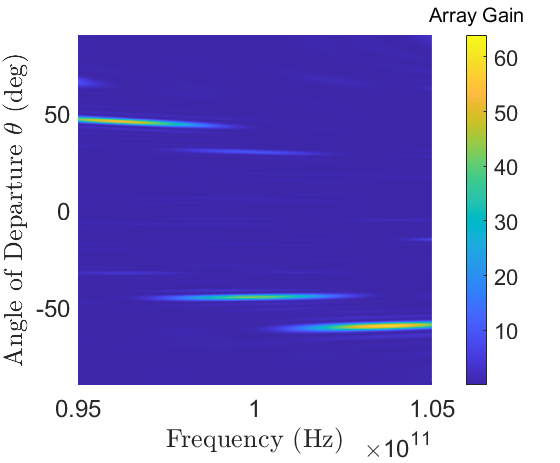} \label{Fig_behav3_algo1}} \\
\subfloat[Ideal, Behavior 4]{\includegraphics[width= 0.24\textwidth]{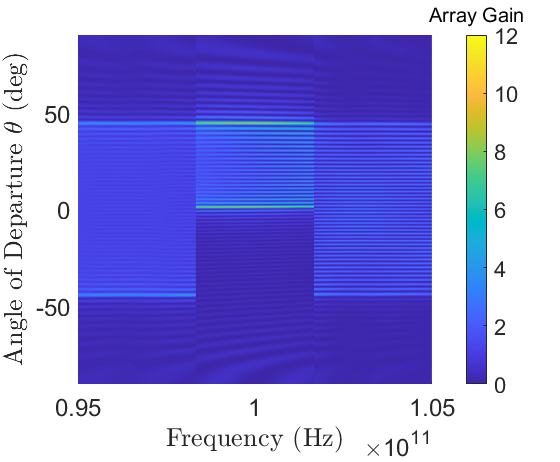} \label{Fig_behav4_ideal}}
\subfloat[JPTA, Behavior 4]{\includegraphics[width= 0.24\textwidth]{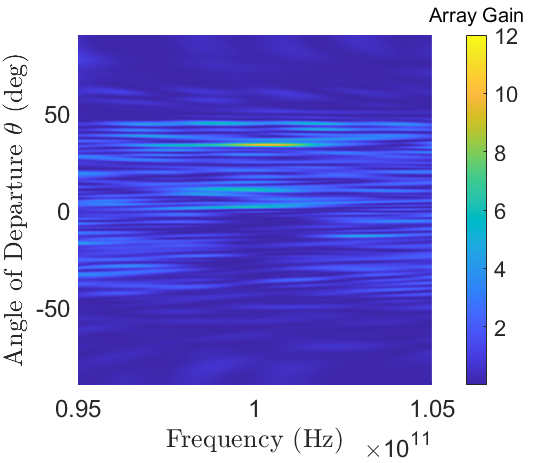} \label{Fig_behav4_algo1}}
\caption{An illustration of the ability of JPTA to form other complex beam behaviors. Here we use Algorithm \ref{Algo1} with \eqref{eqn_opt_TTD} and $\text{max-iter} = 10$.}
\label{Fig_unorthodox_beam_behav}
\end{figure}

\section{Conclusions} \label{sec_conclude}
In this work we propose a new type of hybrid architecture called joint phase-time arrays (JPTA) that can realize frequency-dependent beamforming in the analog domain. We also identify two important frequency dependent beam behaviors 1 \& 2, and listed their various use-cases. It was also shown that the JPTA beamformer design problem can be formulated as an iterative weighted least-squares minimization problem and an algorithm was proposed to solve it. Simulations have shown that JPTA can replicate beam behavior 1 more accurately than behavior 2 and the proposed algorithm can beat other baseline heuristics. We also observe that the proposed iterative algorithm typically converges within $10$ iterations. From the analysis and results, we also infer that for JPTA, the maximum delay to replicate behavior 1 grows linearly as $M/W$ while that of behavior 2 is a constant set to $3/W$. In addition, beam behavior 2 is more robust to reduction in the number of TTDs than behavior 1. We also infer that for conventional hybrid beamforming with partially connected structure, approximately ${\mathrm{O}}(M)$ RF chains are required to replicate behavior 1 and 2, respectively, while JPTA only requires one RF chain. Finally, we also conclude that JPTA can have many different architectures, different optimization objectives and can potentially realize other complex beam behaviors, which can unleash the untapped potential of this newly proposed architecture. 

\begin{appendices}
\section{} \label{appdix1}
\begin{proof}[Proof of Lemma \ref{lemma_wLS}]
Here we make the assumption that for given digital phases $\angle\boldsymbol{\alpha}$, with the optimal parameters $\tau_n^{\circ}(\angle\boldsymbol{\alpha}), \phi_m^{\circ}(\angle\boldsymbol{\alpha})$ the JPTA beamformer can realize a close approximation (modulo $2 \pi$) to the phase angles associated with the desired beam set, i.e., $\angle {[\bar{\mathbf{b}}_k]}_m$ for all $1 \leq m \leq M, k \in \mathcal{K}$. This implies that at the optimal solution we have for all $n \in \{1,...,N\}, m \in \mathcal{M}_n, k \in \mathcal{K}$:
\begin{align}
- 2 \pi f_k \tau_{n} + \phi_m & \stackrel{{\mathrm{mod}} 2}{\approx} \angle {[\bar{\mathbf{b}}_k]}_{m} - \angle \alpha_k \nonumber \\
\Rightarrow - 2 \pi f_k \tau_{n} + \phi_m & \approx \mathscr{U}\Big(\angle {[\bar{\mathbf{b}}_k]}_{m} - \angle\alpha_k \Big), \label{eqn_appdix1_0}
\end{align}
where $\mathscr{U}(\cdot)$ is the phase-unwrapping function as defined in Lemma \ref{lemma_wLS} and the second step follows from the fact that the left hand side: $- 2 \pi f_k \tau_{n} + \phi_m$ can only vary smoothly with $k$. 
Therefore, in the vicinity of the optimal solution, \eqref{eqn_opt_prob_analog_2} can be rewritten as:
\begin{flalign}
& \argmax_{\tau_n, \{\phi_m | m \in \mathcal{M}_n \}} \bigg\{ \sum_{k \in \mathcal{K}} \sum_{m \in \mathcal{M}_n} & \nonumber \\
& \qquad \omega_k \big|{[\bar{\mathbf{b}}_k]}_m \big| {\mathrm{Re}} \Big[ e^{{\mathrm{j}} \big[ \angle\alpha_k - \angle {[\bar{\mathbf{b}}_k]}_m + \phi_{m} - 2 \pi f_k \tau_{n}} \Big] \bigg\} & \nonumber \\
& \approx \argmax_{\tau_n, \{\phi_m | m \in \mathcal{M}_n \}} \bigg\{ \sum_{k \in \mathcal{K}} \sum_{m \in \mathcal{M}_n} \omega_k \big|{[\bar{\mathbf{b}}_k]}_m \big|  & \nonumber \\
& \quad \left[ 1 - \Big( \phi_{m} - 2 \pi f_k \tau_{n} - \mathscr{U}\Big(\angle {[\bar{\mathbf{b}}_k]}_{m} - \angle \alpha_k \Big)^2 \right] \bigg\}, \!\!\!\!\! & \label{eqn_appdix1_1}
\end{flalign}
where \eqref{eqn_appdix1_1} follows by taking the 2nd order Taylor expansion of the $e^{x}$ term, which is accurate near the optimal solution as per \eqref{eqn_appdix1_0}. For any $\angle\boldsymbol{\alpha}$, \eqref{eqn_appdix1_1} is a bounded wLS problem in $|\mathcal{M}_{n}|+1$ variables: $\tau_n, \{\phi_m | m \in \mathcal{M}_n\}$, weights $\omega_k|{[\bar{\mathbf{b}}_k]}_m|$ and a bound $|\tau_n| \leq \kappa/(2W)$. Due to the fact that $\tau_n$ and $\tau_n+K/W$ yield the value of \eqref{eqn_opt_TTD}, and from the concavity of the objective in \eqref{eqn_appdix1_1}, it follows that this bounded wLS can be solved by solving the unconstrained wLS problem and then applying \eqref{eqn_wLS_bounds}. Ignoring the constant terms in \eqref{eqn_appdix1_1} we obtain \eqref{eqn_lemma_wLS}. Note that although $\{\phi_m | m \in \mathcal{M}_n \}$ can also be obtained as a solution of \eqref{eqn_appdix1_1}, we still use prefer using \eqref{eqn_opt_phi} in Lemma \ref{lemma_wLS} since \eqref{eqn_opt_phi} is in closed form and doesn't depend on the Taylor approximation.
\end{proof}

\section{} \label{appdix2}
\begin{proof}[Proof of Proposition \ref{remark_beamsway_FC_PC}]
We first define a kind of $M \times 1$ vector
\begin{align}
\mathbf{g}(\Omega) \triangleq [1 \quad e^{\mathrm{j} \pi \Omega} \quad \cdots \quad e^{\mathrm{j} \pi \Omega (M-1)}]^{\rm T}.
\end{align}
It can be easily shown that when $\left|\Omega_1 - \Omega_2 \right| = \frac{2k}{M}$ where $k$ is a non-zero integer,  $\mathbf{g}(\Omega_1)^{\dagger} \mathbf{g}(\Omega_2)=0$ which means that $\mathbf{g}(\Omega_1)$ is orthogonal to $\mathbf{g}(\Omega_2)$. 

In the matrix $\mathbf{B}$, the first column is equal to $\mathbf{g} \left( \sin \left(\theta_0- \frac{\Delta \theta}{2} \right)  \frac{f_{\lfloor \frac{1-K}{2} \rfloor}}{f_c} \right)$ and the last column is equal to $\mathbf{g} \left( \sin \left(\theta_0 + \frac{\Delta \theta}{2} \right)  \frac{f_{\lfloor \frac{K-1}{2} \rfloor}}{f_c} \right)$. Therefore, there are $r$ columns in the $M \times K$ matrix $\mathbf{B}$ who are orthogonal to each other.
We assume that $K \gg M$ in this paper (e.g., $K=2048$, $M=64$). The rank of $\mathbf{B}$ is thus no less than $r$. 

On the other side, the rank of the hybrid beamforming $\FRF \FBB$ is restricted by $N_{\mathrm{RF}}$, since $K \gg M \geq N_{\mathrm{RF}}$ and $\operatorname{rank}(\mathbf{P} \mathbf{Q}) \leq \min(\operatorname{rank}(\mathbf{P}), \operatorname{rank}(\mathbf{Q}))$. To approximate well the rank-$r$ (or $\geq r$) matrix $\mathbf{B}$, $N_{\mathrm{RF}}$ has to be no less than $r$. Therefore, we obtain the lower bound of $N_{\mathrm{RF}}$ for HBF-FC in the proposition.

For the case of HBF-PC, the structure limitation only increases the difficulties of the replication, and more RF chains are therefore needed. In addition, we assume that the number of RF in HBF-PC is $2^q$ where $q\geq 1$ is an integer.

Lastly, note that when the bandwidth is small, i.e., $\frac{f_{\lfloor \frac{K-1}{2} \rfloor}}{f_0} \approx 1$ and $\frac{f_{\lfloor \frac{1-K}{2} \rfloor}}{f_0} \approx 1$, an approximation of $r$ is,
\begin{align}
r \approx \left \lceil \frac{M}{2} \left| \sin \left(\theta_0+\frac{\Delta \theta}{2} \right) - \sin \left(\theta_0-\frac{\Delta \theta}{2} \right) \right| \right \rceil.
\end{align}
\end{proof}

\end{appendices}

\section*{Acknowledgments}
The authors would like to thank the constructive inputs from Jin Yuan, Shadi Abu-Surra, and Gary Xu of Samsung Research America on the hardware implementation of JPTA and the comparison to HBF.





%

\bibliographystyle{IEEEtran}
\bibliography{references}

\begin{IEEEbiography}[{\includegraphics[width=1in,height=1.25in,clip,keepaspectratio]{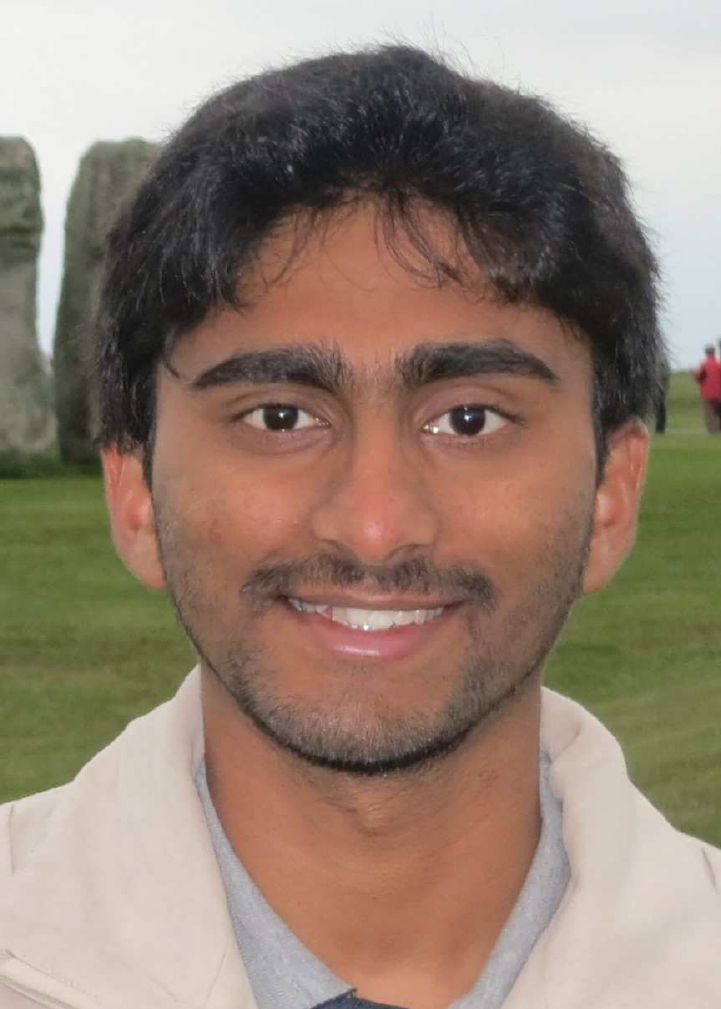}}]{Vishnu V. Ratnam} (S'10--M'19--SM'22) received the B.Tech. degree (Hons.) in electronics and electrical communication engineering from IIT Kharagpur, Kharagpur, India in 2012, where he graduated as the Salutatorian for the class of 2012. He received the Ph.D. degree in electrical engineering from University of Southern California, Los Angeles, CA, USA in 2018. He currently works as a staff research engineer in the Standards and Mobility Innovation Lab at Samsung Research America, Plano, Texas, USA. His research interests are in AI for wireless, mm-Wave and Terahertz communication, massive MIMO, wireless sensing, and resource allocation problems in multi-antenna networks.
Dr. Ratnam was the recipient of the Best Student Paper Award with the IEEE International Conference on Ubiquitous Wireless Broadband (ICUWB) in 2016, the Bigyan Sinha memorial award in 2012 and is a member of the Phi-Kappa-Phi honor society.
\end{IEEEbiography}

\begin{IEEEbiography}[{\includegraphics[width=1in,height=1.25in,clip,keepaspectratio]{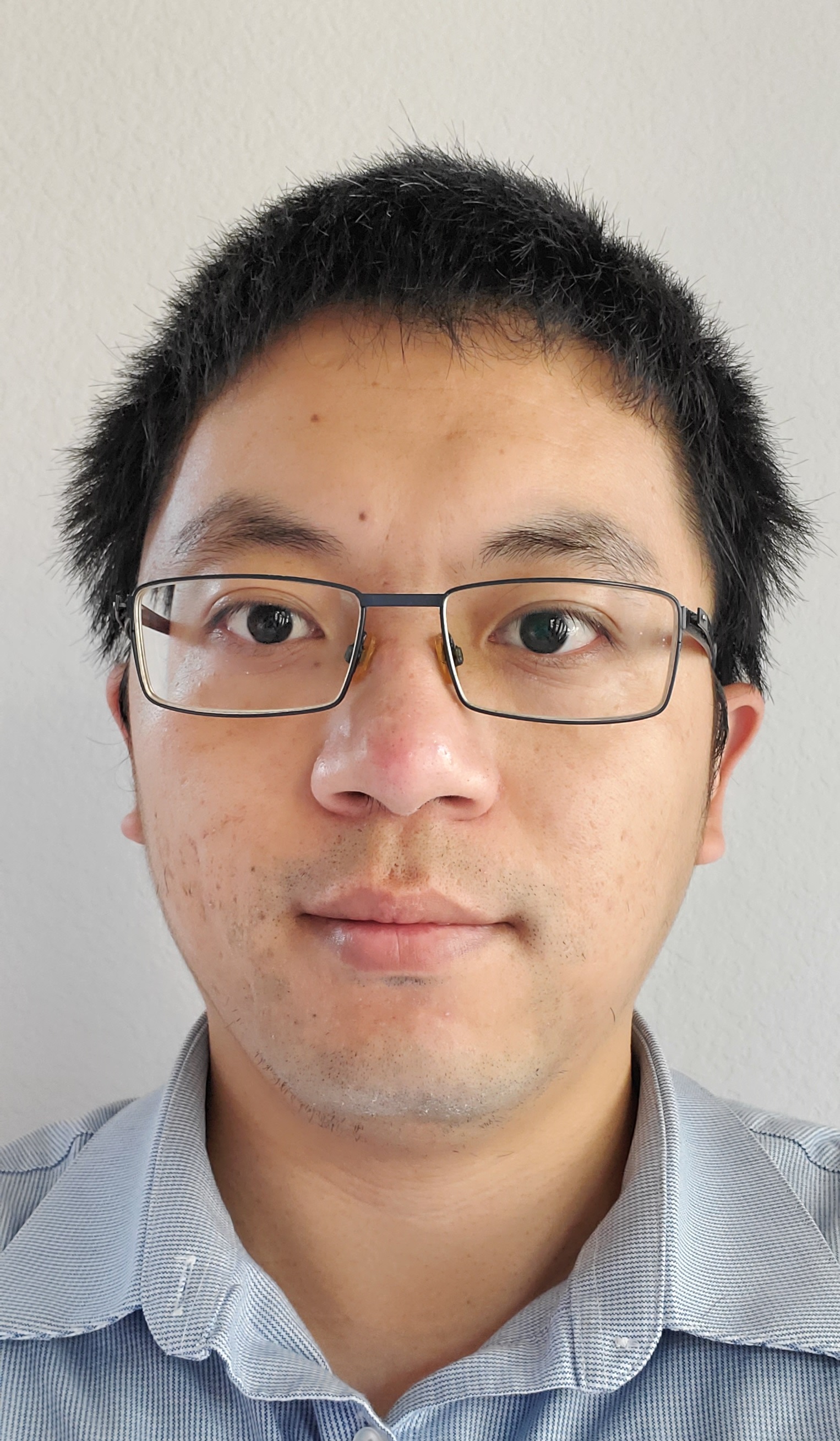}}]{Jianhua Mo} (S'12--M'17--SM'22) received the B.S. and M.S. degrees from Shanghai Jiao Tong University in 2010 and 2013, respectively, and the Ph.D. degree from The University of Texas at Austin in 2017, all in electronic engineering. 
He is currently a Staff Engineer with the Samsung Research America, Plano, TX, USA. His areas of interest includes physical layer security, MIMO communications with low resolution ADCs, and mmWave beam codebook design and beam management. 
Dr. Mo's awards and honors include Heinrich Hertz Award of 2013, Stephen O. Rice Prize of 2019, ``Best Wi-Fi Innovation Award'' by  Wireless Broadband Alliance (WBA) in 2019, Exemplary Reviewer of the IEEE Wireless Communications Letters in 2012, Exemplary Reviewer of the IEEE Communications Letters in 2015, and Finalist for Qualcomm Innovation Fellowship in 2014.
\end{IEEEbiography}

\begin{IEEEbiography}[{\includegraphics[width=1in,height=1.25in,clip,keepaspectratio]{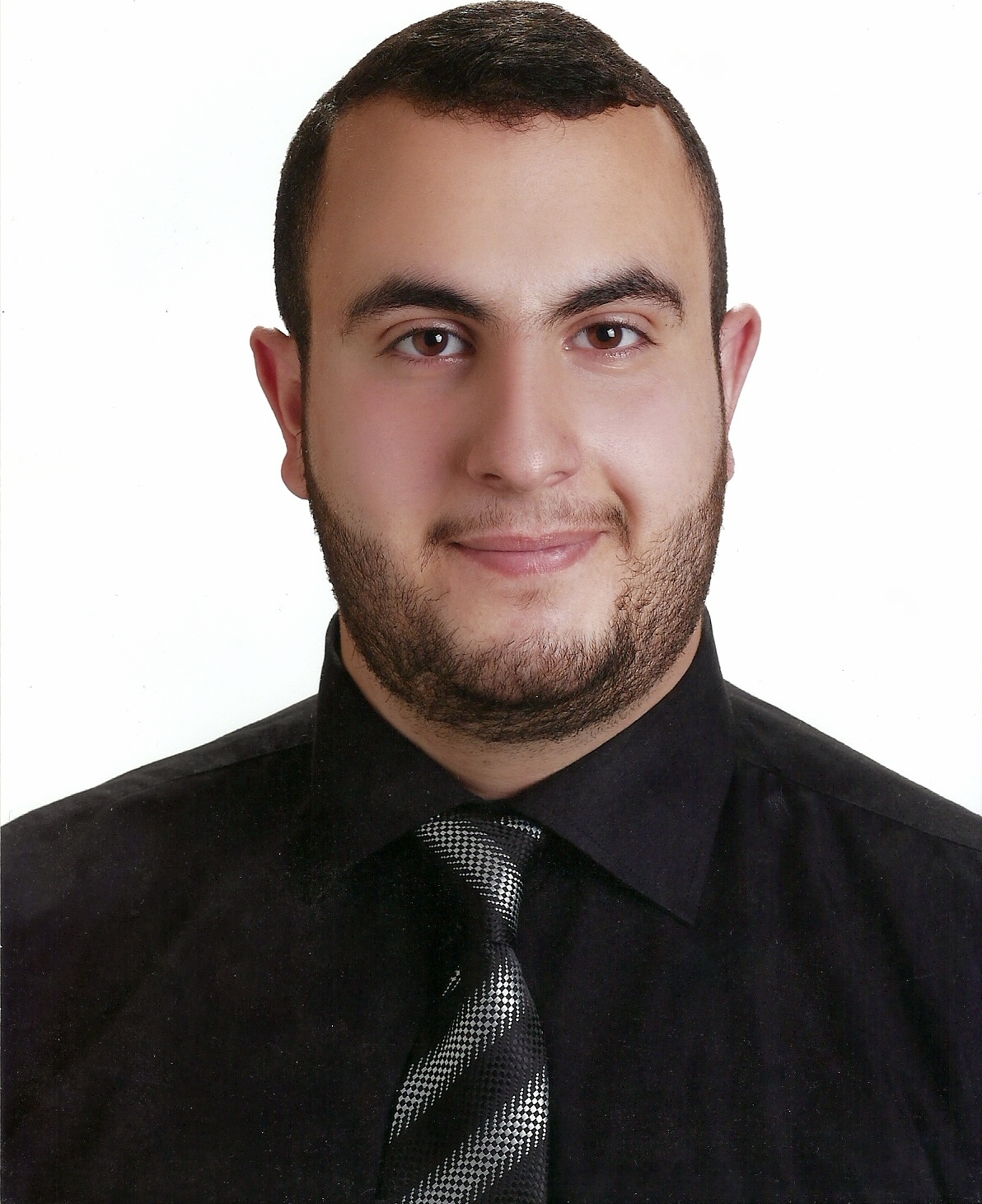}}]{Ahmad AlAmmouri} (S'11--M'21) received his B.Sc. degree from the University of Jordan, Amman, Jordan, in 2014, his M.Sc.\ degree from King Abdullah University of Science and Technology (KAUST), Thuwal, Saudi Arabia, in 2016, and his PhD degree from The University of Texas at Austin, Texas, in 2020, all in Electrical Engineering. He is currently a Senior Research Engineer with the Samsung Research America, Plano, TX, USA.  He has held summer internships at Samsung Research America, Richardson, TX, in 2017 and 2018, and was a visiting researcher at INRIA, Paris, in 2019 and 2020. He was awarded the Chateaubriand Fellowship from the French Embassy in the USA and the Professional Development Award from UT Austin, both in 2019, and the WNCG Student Leadership Award in 2020. He was recognized as an Exemplary Reviewer by the IEEE Transactions on Communications in 2017 and by IEEE Transactions on Wireless Communications in 2017 and 2018.
\end{IEEEbiography}

\begin{IEEEbiography}[{\includegraphics[width=1in,height=1.25in,clip,keepaspectratio]{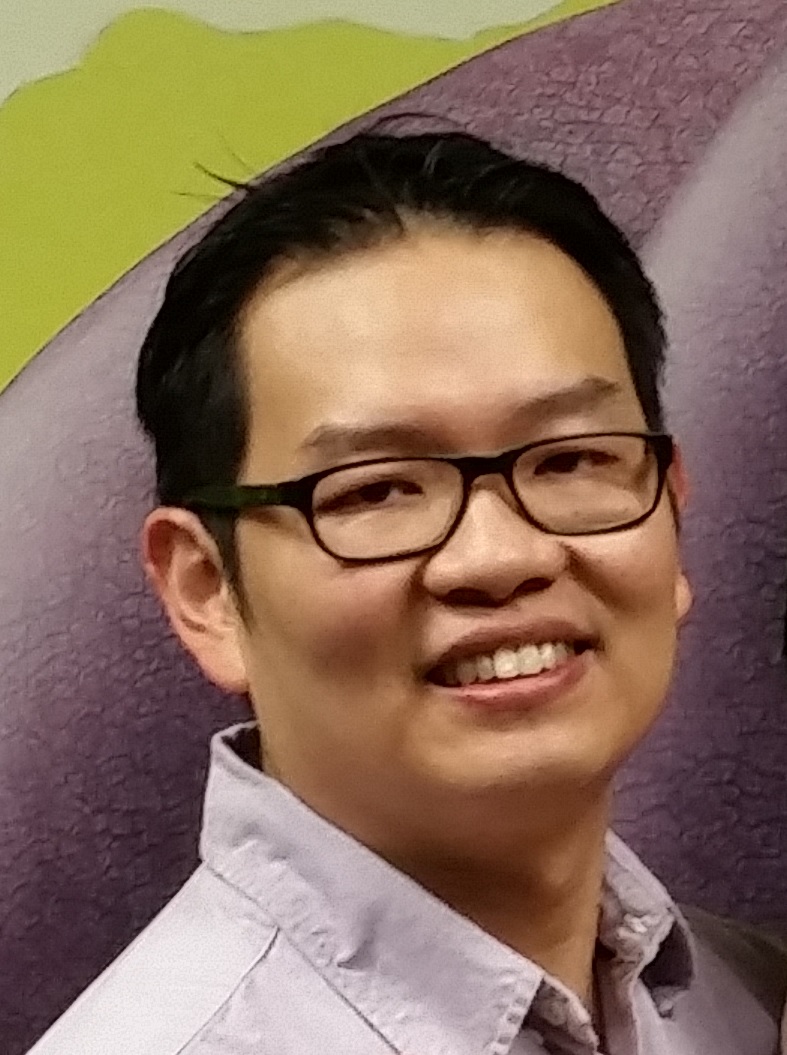}}]{Boon Loong Ng} (S'03-M'08) received the B. Eng. degree in electrical and electronic engineering and the Ph.D. degree in engineering from the University of Melbourne, Australia in 2001 and 2007, respectively.
He is currently a Senior Research Director with the Standards and Mobility Innovation (SMI) Laboratory, Samsung Research America, Plan, TX, USA. He had contributed to 3GPP RAN L1/L2 standardization of LTE, LTE-A, and 5G NR technologies from the period of 2008 to 2018. He holds over 60 USPTO-granted patents on LTE/LTE-A/5G and more than 100 patent applications globally. Since 2018, he has been leading a research and development team that develops system and algorithm design solutions for commercial 5G, Wi-Fi and UWB technologies.
\end{IEEEbiography}

\begin{IEEEbiography}[{\includegraphics[width=1in,height=1.25in,clip,keepaspectratio]{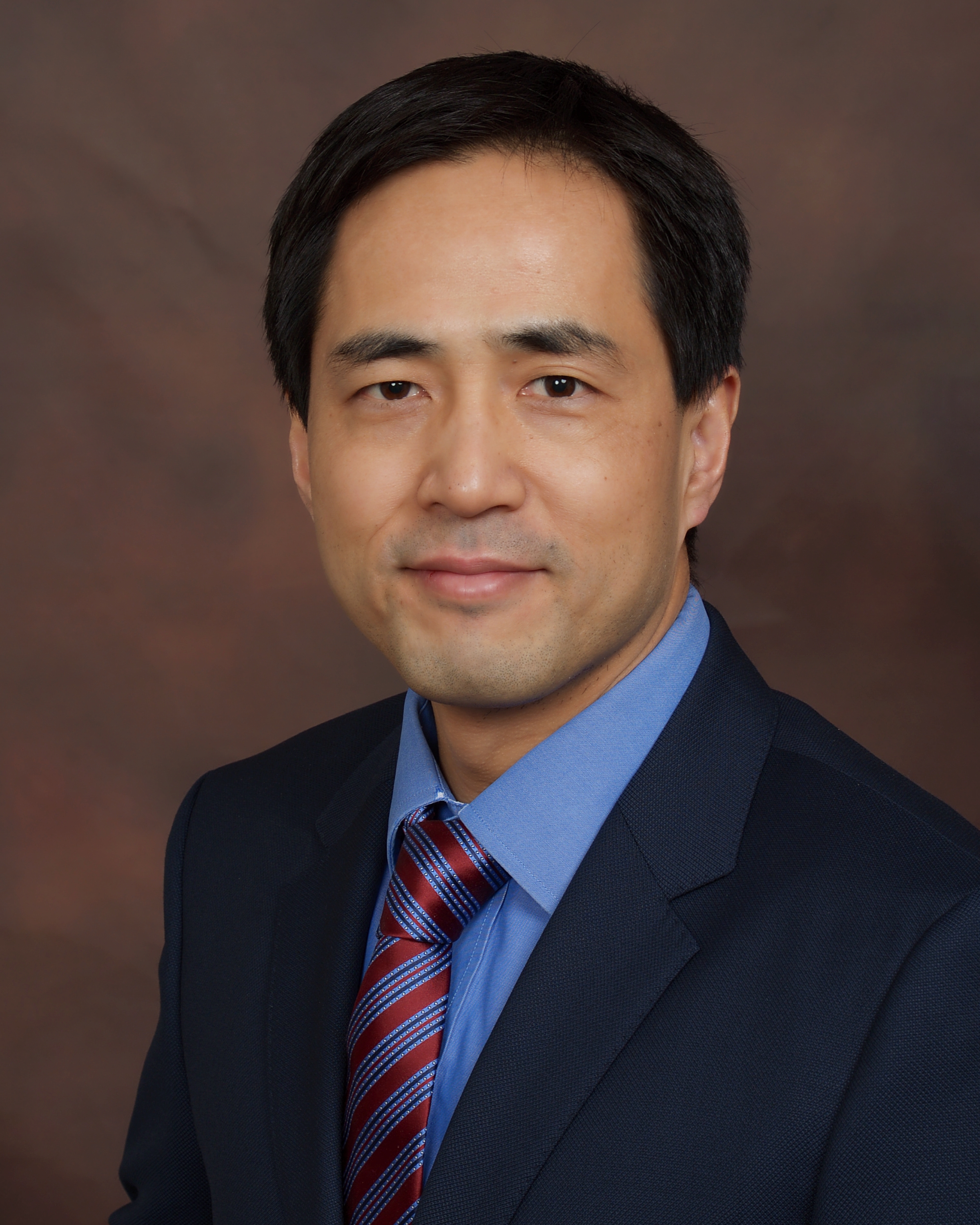}}]{Jianzhong (Charlie) Zhang} (S'00-M'03-SM'09-F'16) received the Ph.D. degree from the University of Wisconsin, Madison WI, USA. He was with the Nokia Research Center from 2001 to 2006, where he was involved in the IEEE 802.16e (WiMAX) standard and EDGE/CDMA receivers, and from 2006 to 2007, he was with Motorola, where he was involved in 3GPP HSPA standards. From 2009 to 2013, he served as the Vice Chairman of the 3GPP RAN1 group and led the development of LTE and LTE-Advanced technologies, such as 3-D channel modeling, UL-MIMO and CoMP, and Carrier Aggregation for TD-LTE. He is currently the Senior Vice President and the Head of the Standards and Mobility Innovation Laboratory, Samsung Research America, where he leads research, prototyping, and standards for 5G cellular systems and future multimedia networks. 
\end{IEEEbiography}

\begin{IEEEbiography}[{\includegraphics[width=1in,height=1.25in,clip,keepaspectratio]{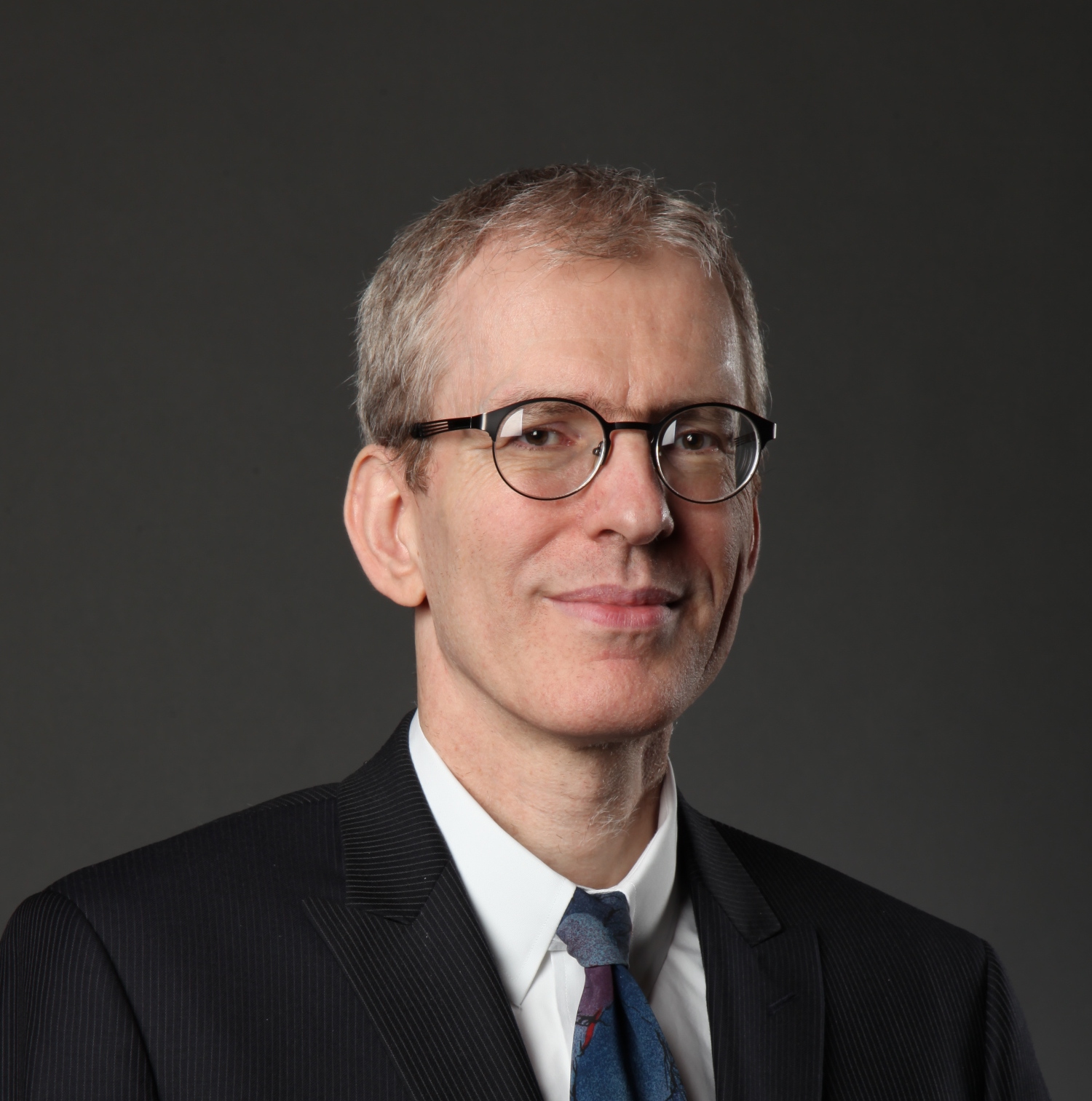}}]{Andreas F. Molisch} (S'89-M'95-SM'00-F'05) received his degrees (Dipl.Ing. 1990, PhD 1994, Habilitation 1999) from the Technical University Vienna, Austria. He spent the next 10 years in industry, at FTW, AT\&T (Bell) Laboratories, and Mitsubishi Electric Research Labs (where he rose to Chief Wireless Standards Architect). In 2009 he joined the University of Southern California (USC) in Los Angeles, CA, as Professor, and founded the Wireless Devices and Systems (WiDeS) group. In 2017, he was appointed to the Solomon Golomb – Andrew and Erna Viterbi Chair. 
His research interests revolve around wireless propagation channels, wireless systems design, and their interaction. Recently, his main interests have been wireless channel measurement and modeling for 5G and beyond 5G systems, joint communication-caching-computation, hybrid beamforming, UWB/TOA based localization, and novel modulation/multiple access methods. Overall, he has published 4 books (among them the textbook “Wireless Communications”, currently in its second edition), 21 book chapters, 280 journal papers, and 370 conference papers. He is also the inventor of 70 granted (and more than 10 pending) patents, and co-author of some 70 standards contributions. His work has been cited more than 58,000 times, his h-index is >100, and he is a Clarivate Highly Cited Researcher. 

Dr. Molisch has been an Editor of a number of journals and special issues, General Chair, Technical Program Committee Chair, or Symposium Chair of multiple international conferences, as well as Chairperson of various international standardization groups. He is a Fellow of the National Academy of Inventors, Fellow of the AAAS, Fellow of the IEEE, Fellow of the IET, an IEEE Distinguished Lecturer, and a member of the Austrian Academy of Sciences. He has received numerous awards, among them the IET Achievement Medal, the Technical Achievement Awards of IEEE Vehicular Technology Society (Evans Avant-Garde Award) and the IEEE Communications Society (Edwin Howard Armstrong Award), and the Technical Field Award of the IEEE for Communications, the Eric Sumner Award.
\end{IEEEbiography}

\end{document}